\documentclass[a4paper,UKenglish,cleveref,autoref,thm-restate]{lipics-v2021}
\nolinenumbers

\title{On the Complexity of  {Intersection Non-emptiness} for Star-Free Language Classes} 

\author{Emmanuel Arrighi}{University of Bergen, Norway}{emmanuel.arrighi@uib.no}{https://orcid.org/0000-0002-0326-1893}{Research Council of Norway (no. 274526), IS-DAAD (no. 309319)
}

\author{Henning Fernau}
{Universit\"at Trier, Fachbereich IV, Informatikwissenschaften, Germany \and \url{https://www.uni-trier.de/index.php?id=49861}}
{fernau@informatik.uni-trier.de}
{https://orcid.org/0000-0002-4444-3220}
{DAAD PPP (no. 57525246)}

\author{Stefan Hoffmann}{Universit\"at Trier, Fachbereich IV, Informatikwissenschaften, Germany}{hoffmanns@informatik.uni-trier.de}{https://orcid.org/0000-0002-7866-075X}{}

\author{Markus Holzer}{Universit\"at Giessen, Institut f\"ur Informatik, 
Germany}{holzer@informatik.uni-giessen.de}{https://orcid.org/0000-0003-4224-4014}{}

\author{Isma\"el Jecker}{Institute of Science and Technology, Klosterneuburg, Austria}{ismael.jecker@ist.ac.at}{https://orcid.org/0000-0002-6527-4470}{Marie Skłodowska-Curie Grant Agreement no. 754411}

\author{Mateus de Oliveira Oliveira}{University of Bergen, Norway}{mateus.oliveira@uib.no}{https://orcid.org/0000-0001-7798-7446}{Research Council of Norway (no. 288761), IS-DAAD (no. 309319)}

\author{Petra Wolf}
{Universit\"at Trier, Fachbereich IV, Informatikwissenschaften, Germany \and \url{https://www.wolfp.net/}}
{wolfp@informatik.uni-trier.de}
{https://orcid.org/0000-0003-3097-3906}
{DFG project FE 560/9-1, DAAD PPP (no. 57525246)}

\authorrunning{E. Arrighi et al.}

\Copyright{E. Arrighi et al.} 

\ccsdesc[500]{Theory of computation~Regular languages}
\ccsdesc[500]{Theory of computation~Problems, reductions and completeness}

\keywords{Intersection Non-emptiness Problem, Star-Free Languages, Straubing-Th\'{e}rien Hierarchy, dot-depth Hierarchy, Commutative Languages, Complexity} 

\acknowledgements{We like to thank Lukas Fleischer and Michael Wehar for our discussions. This work started at the Schloss Dagstuhl Event 20483 \emph{Moderne Aspekte der Komplexitätstheorie in der Automatentheorie}~\url{https://www.dagstuhl.de/20483}.}

\EventEditors{John Q. Open and Joan R. Access}
\EventNoEds{2}
\EventLongTitle{42nd Conference on Very Important Topics (CVIT 2016)}
\EventShortTitle{CVIT 2016}
\EventAcronym{CVIT}
\EventYear{2016}
\EventDate{December 24--27, 2016}
\EventLocation{Little Whinging, United Kingdom}
\EventLogo{}
\SeriesVolume{42}
\ArticleNo{23}
\usepackage[nocompress]{cite}
\usepackage{comment}
\usepackage[bibliography=separate]{apxproof}
\usepackage{amsmath}
\usepackage{amsfonts}
\usepackage{complexity}
\usepackage{verbatim}
\usepackage{color}
\usepackage{url}
\usepackage{xspace}
\usepackage[color=gray!27]{todonotes}
\usepackage{shuffle}
\DeclareFontFamily{U}{shuffle}{}
\DeclareFontShape{U}{shuffle}{m}{n}{ <-8>shuffle7 <8->shuffle10}{}
\DeclareMathOperator{\lcm}{lcm}
\usepackage{tikz}
\usetikzlibrary{automata, positioning, arrows}
\usetikzlibrary{fadings,patterns}
\usetikzlibrary{intersections}

\pgfdeclarelayer{background}
\pgfsetlayers{background,main}

\newcommand{\resp}{respectively}

\newcommand{\comm}{\texttt{perm}}

\newcommand{\dd}[1]{{\mathcal B}_{#1}}

\renewcommand{\st}[1]{{\mathcal L}_{#1}}

\newcommand{\FANEI}{\textsc{FA-NEI}\xspace}
\newcommand{\INE}{\textsc{Intersection Non-emptiness}\xspace}

\begin{document}

\maketitle
        
\begin{abstract}
  In the \INE\ problem, we are given a list of finite automata
  $A_1,A_2,\dots,A_m$ over a common alphabet~$\Sigma$ as input, and
  the goal is to determine whether some string $w\in \Sigma^*$ lies in
  the intersection of the languages accepted by the automata in the
  list.  We analyze the complexity of the \INE problem under the
  promise that all input automata accept a language in some level of
  the dot-depth hierarchy, or some level of the Straubing-Th\'{e}rien
  hierarchy. Automata accepting languages from the lowest levels of
  these hierarchies arise naturally in the context of model checking.
  We identify a dichotomy in the dot-depth hierarchy by showing that
  the problem is already \NP-complete when all input automata accept
  languages of the levels~$\dd{0}$ or~$\dd{1/2}$ and already
  \PSPACE-hard when all automata accept a language from the
  level~$\dd{1}$. Conversely, we identify a tetrachotomy in the
  Straubing-Th\'{e}rien hierarchy. More precisely, we show that the
  problem is in~$\AC^0$ when restricted to level~$\st{0}$; complete
  for $\L$ or $\NL$, depending on the input representation, when
  restricted to languages in the level~$\st{1/2}$; \NP-complete when
  the input is given as DFAs accepting a language in $\st{1}$ or
  $\st{3/2}$; and finally, \PSPACE-complete when the input automata
  accept languages in level~$\st{2}$ or higher. Moreover, we show that
  the proof technique used to show containment in $\NP$ for DFAs
  accepting languages in $\st{1}$ or $\st{3/2}$ does not generalize to
  the context of NFAs. To prove this, we identify a family of
  languages that provide an exponential separation between the state
  complexity of general NFAs and that of partially ordered NFAs. To
  the best of our knowledge, this is the first superpolynomial
  separation between these two models of computation.
\end{abstract}

\section{Introduction}
The \INE problem for finite automata is one of the most fundamental and well studied 
problems in the interplay between algorithms, complexity theory, and automata theory 
\cite{Kozen1977Lower,ie85,ie92,ie01,Lipton2003Intersection,Wehar2014hardness,FernauK17,Wehar2016thesis}.
Given a list $A_1,A_2,\dots,A_m$ of finite automata over a common alphabet $\Sigma$, 
the goal is to determine whether there is a string $w\in \Sigma^*$ that is accepted 
by each of the automata in the list. This problem is $\PSPACE$-complete when no restrictions 
are imposed \cite{Kozen1977Lower}, and becomes \NP-complete when the input automata accept
unary languages (implicitly contained already in~\cite{StMe73}) or 
finite languages~\cite{Rampersad2010detecting}. 

In this work, we analyze the complexity of the \INE problem under the assumption that the languages accepted by 
the input automata belong to a given level of the Straubing-Th\'erien hierarchy \cite{place2017concatenation,St81,straubing1985finite,Th81} or 
to some level of the Cohen-Brzozowski dot-depth hierarchy \cite{Brz76,CoBr71,place2017concatenation}. Somehow, these languages are 
severely restricted, in the sense that both hierarchies, which are infinite, 
are entirely contained in the class of star-free languages, a class of languages 
 that can be represented by expressions that use union, concatenation, and complementation, but \emph{no} Kleene star operation \cite{Brz76,brzozowski1978dot,place2017concatenation}.
Yet, languages belonging to fixed levels of either hierarchy may already be 
very difficult to characterize, in the sense that the very problem of deciding whether the language accepted by 
a given finite automaton belongs to a given full level or half-level $k$ of either hierarchy is open, except 
for a few values of $k$ \cite{klima2010new,glasser2001level,gla2000decidable,place2017concatenation}. 
It is worth noting that while the problem of determining whether a given automaton accepts a language 
in a certain level of either the dot-depth or of the Straubing-Th\'{e}rien hierarchy is computationally hard (Theorem~\ref{thm:membership-pspace-hardness}),
automata accepting languages in lower levels of these hierarchies arise naturally in a variety of applications 
such as model checking where the \INE problem is of fundamental relevance \cite{Abdulla12,BouajjaniJNT00,BMT07}.

An interesting question to consider is how the complexity of the \INE problem changes as we 
move up in the levels of the Straubing-Th\'{e}rien hierarchy or in the levels of the dot-depth hierarchy. 
In particular, does the complexity of this problem changes gradually, as we increase the complexity 
of the input languages? In this work, we show that this is actually not the case, and that the 
complexity landscape for the \INE problem is already determined by the very first levels of 
either hierarchy (see \autoref{fig:diagram}).
%
%
%
Our first main result states that the \INE{} problem for NFAs and DFAs
accepting languages from the level $1/2$ of the Straubing-Th\'erien hierarchy are $\NL$-complete and $\L$-complete,
respectively, under $\AC^0$ reductions (\autoref{thm:NL-completeness}). Additionally, this completeness 
result holds even in the case of unary languages. To prove hardness for $\NL$ and $\L$, respectively, 
we will use a simple reduction from the reachability problem for DAGs and for directed trees, respectively. 
Nevertheless, the proof of containment in $\NL$ and in $\L$, respectively, will require a new insight that
may be of independent interest. More precisely, we will use a characterization of languages in the level 
$1/2$ of the Straubing-Th\'{e}rien hierarchy as shuffle ideals to show 
that the \INE{} problem can be reduced to {\sc concatenation non-emptiness}
(\autoref{lem:half-level-intersection-nonemptiness-iff-emptiness}). This  allows us to decide 
\INE{} by analyzing each finite automaton given at the input individually. 
It is worth mentioning that this result is optimal in the sense that the problem becomes \NP-hard  
even if we allow a single DFA to accept a language from $\st{1}$, and require all the others to accept languages
from $\st{1/2}$ (\autoref{thm:ST-1-NP-hard}). 

Subsequently, we analyze the complexity of \INE{} when all input automata
are assumed to accept languages from one of the levels of  $\dd{0}$ or $\dd{1/2}$ of the dot-depth hierarchy, 
or from the levels $\st{1}$ or $\st{3/2}$ of the Straubing-Th\'{e}rien hierarchy.
It is worth noting that $\NP$-hardness follows straightforwardly from the fact that 
\INE for DFAs accepting finite languages is already \NP-hard \cite{Rampersad2010detecting}. Containment in \NP, on the other hand,
is a more delicate issue, and here the representation of the input automaton plays an important
 role. A characterization of languages in $\st{3/2}$ 
in terms of languages accepted by  partially ordered NFAs~\cite{STV01} is crucial for us,
combined with the fact that \INE{} when the input is given by such automata 
is \NP-complete \cite{MaTh15}. Intuitively, the proof in \cite{MaTh15} 
follows by showing that the minimum 
length of a word in the intersection of languages in the level 
$3/2$ of the Straubing-Th\'{e}rien hierarchy is bounded by a polynomial on the sizes of the 
minimum partially ordered NFAs accepting these languages. To prove that \INE{} 
is in $\NP$ when the input automata are given as DFAs, we prove a new result establishing that the number 
of Myhill-Nerode equivalence classes in a language in the level $\st{3/2}$ is at least as large as the 
number of states in a minimum partially ordered automaton representing the same language (\autoref{boundSizepoNFA}). 

Interestingly, we show that the proof technique used to prove this last result does not generalize to the context of NFAs. 
To prove this, we carefully design a sequence ${(L_n)}_{n\in \mathbb{N}_{\geq 1}}$ of languages over a binary alphabet such that for every $n\in \mathbb{N}_{\geq 1}$,  
the language $L_n$ can be accepted by an NFA of size $n$, but any partially ordered NFA accepting $L_n$ has size $2^{\Omega(\sqrt{n})}$.
This lower bound is ensured by the fact that the syntactic monoid of $L_n$ has many $\mathcal{J}$-factors. 
Our construction is inspired by a technique introduced by Klein and Zimmermann, in a completely different context,
to prove lower bounds on the amount of look-ahead necessary to win infinite games with delay~\cite{KleinZ14}. 
%
%
%
%
%
%
%
%
%
To the best of our knowledge, this is the first exponential separation between 
the state complexity of general NFAs and that of partially ordered NFAs.
While this result does not exclude the possibility that 
\INE for languages in $\st{3/2}$ represented by general NFAs is in~\NP, it gives some indication 
that proving such a containment  requires substantially new techniques. 

Finally, we show that \INE{} for both DFAs and for NFAs is already $\PSPACE$-complete if all 
accepting languages are from the level $\dd{1}$ of the dot-depth hierarchy or from the level
$\st{2}$ of the Straubing-Th\'{e}rien hierarchy. 
We can adapt   Kozen's classical \PSPACE-completeness proof by using the complement of languages introduced in~\cite{DBLP:journals/corr/abs-1907-13115} 
in the study of partially ordered automata. Since the languages in \cite{DBLP:journals/corr/abs-1907-13115} belong to 
$\st{3/2}$, their complement belong to $\st{2}$ (and to $\dd{1}$), and therefore, the proof follows.

\section{Preliminaries} 
We let $\mathbb{N}_{\geq k}$ denote the set of natural numbers greater or equal than $k$.

We assume the reader to be familiar with the basics in computational
complexity theory~\cite{Pa94a}. In particular, we recall the inclusion
chain:
$\AC^0\subset\NC^1\subseteq\L\subseteq\NL\subseteq\P\subseteq
\NP\subseteq\PSPACE$.
Let~$\AC^0$ ($\NC^1$, \resp) refer to the class of
problems accepted by Turing machines with a bounded (unbounded, \resp)
number of alternations in logarithmic time; alternatively one can
define these classes by uniform Boolean circuits.  Here, \L\ (\NL,
\resp) refers to the class of problems that are accepted by
deterministic (nondeterministic, \resp) Turing machines with
logarithmic space, \P\ (\NP, \resp) denotes the class of problems
solvable by deterministic (nondeterministic, \resp) Turing machines in
polynomial time, and \PSPACE\ refers to the class of languages
accepted by deterministic or nondeterministic Turing machines in
polynomial space~\cite{Sa70a}. Completeness and hardness are always
meant with respect to deterministic logspace many-one reductions
unless otherwise stated. We will also consider the parameterized class
\XP{} of problems that can be solved in time $n^{f(k)}$, where $n$ is the size
of the input, $k$ is a parameter, and $f$ is a computable function \cite{FluGro2006}.

We  mostly consider \emph{nondeterministic finite
  automata} (NFAs).  An NFA~$A$ is a tuple $A = (Q,
\Sigma, \delta, q_0, F)$, where~$Q$ is the finite \emph{state set}
with the \emph{start state}~$q_0\in Q$, the \emph{alphabet}~$\Sigma$
is a finite set of input symbols, and $F \subseteq Q$ is the
\emph{final state set}.  The \emph{transition function} $\delta :
Q\times \Sigma \to 2^Q$ extends to words from $\Sigma^*$ as usual. 
Here,~$2^Q$ denotes the powerset of~$Q$.  
By $L(A) = \{\, w \in \Sigma^* \mid \delta(q_0, w) \cap
F\neq\emptyset\,\}$, we denote the \emph{language accepted by}~$A$.
The NFA~$A$ is a \emph{deterministic finite automaton} (DFA) if
$|\delta(q,a)|=1$ for every $q\in Q$ and $a\in\Sigma$. Then, 
we simply write $\delta(q,a)=p$ instead of $\delta(q,a)=\{p\}$. If
$|\Sigma| = 1$, we call $A$ a \emph{unary} automaton.

%
%

We study \INE\ problems and their
complexity. For finite automata, this problem is defined as follows:
\begin{itemize}
    \item \emph{Input}: Finite automata $A_i=(Q_i,\Sigma,\delta_i,q_{(0,i)},F_i)$, for $1\leq i\leq m$.
\item \emph{Question}: Is there a word~$w$ that is accepted by all~$A_i$,
i.e., is $\bigcap_{i=1}^m L(A_i)\neq\emptyset$?
\end{itemize}
Observe that the automata have a common input alphabet. Note that the
complexity of the non-emptiness problem for finite automata of a
certain type is a lower bound for the \INE\ for this particular type
of automata.
%
%
Throughout the paper we are mostly interested in the complexity of the \INE\
problem for finite state devices 
whose languages are contained in a particular language class. 
%
%
%


\begin{figure}[t]%
    \centering
    \scalebox{0.9}{
    \begin{tikzpicture}[>=stealth']
    \useasboundingbox (-1,-.75) rectangle (10, 2.5);
    \node at (0,0) (st0) {$\st{0}$};
    \node at (2,0) (st12) {$\st{\frac{1}{2}}$};
    \node at (2, 2) (dd0) {$\dd{0}$};
    \node at (4,0) (st1) {$\st{1}$};
    \node at (4, 2) (dd12) {$\dd{\frac{1}{2}}$};
    \node at (6,0) (st32) {$\st{\frac{3}{2}}$};
    \node at (6, 2) (dd1) {$\dd{1}$};
    \node at (8,0) (st2) {$\st{2}$};

    \draw[->] (st0) to (dd0);
    \draw[->] (st0) to (st12);
    \draw[->] (st12) to (dd12);
    \draw[->] (st12) to (st1);
    \draw[->] (st1) to (dd1);
    \draw[->] (st1) to (st32);
    \draw[->] (st32) to (st2);
    \draw (st32) to +(45:2.12);
    \draw[dotted] (st32) ++(45:2.12) to +(45:.70);
    \draw[dotted] (st2) to +(0:1);

    \draw[->] (dd0) to (dd12);
    \draw[->] (dd0) to (st1);
    \draw[->] (dd12) to (dd1);
    \draw[->] (dd12) to (st32);
    \draw[->] (dd1) to (st2);
    \draw (dd1) to +(0:1.5);
    \draw[dotted] (dd1) ++(0:1.5) to +(0:.5);

    \begin{pgfonlayer}{background}
        \fill[color=green!10] (st0) +(-.5, 1) to[out=30, in=90] +(2.5, 0) .. controls +(-90:2) and +(-150:3) .. cycle;

        \draw[color=green!30,line width=8,line cap=round,name path=nl] (st0) +(-.5, 1) to[out=30, in=90] node[color=black, very near start] [sloped] {\NL} +(2.5, 0) .. controls +(-90:2) and +(-150:3) .. cycle;
        \path[name path=top] (st0) ++(-1, 0.5) to +(-2,0);
        \path[name path=bottom] (st0) ++(1, -0.5) to +(0,-2);
        \path[name intersections={of= nl and top }] (intersection-1) node (top) {};
        \path[name intersections={of= nl and bottom }] (intersection-1) node (bottom) {};
        \draw[color=green!30,line width=2,line cap=round,dashed] (st0) +(-1, 0.75) to[out=0, in=90] +(1, -0.5) to (bottom);
        \draw[color=green!30,line width=2,line cap=round,dashed] (st0) +(-1, 0.75) to (top);

        \fill[color=orange!10] (dd0) ++(-0.7, .5) .. controls +(-70:1) and +(180:1) .. ++(3, -3) to ++(5,0) to ++(0, 3) .. controls +(180:1) and +(110:1) .. cycle;
        \draw[color=orange!30,line width=8,line cap=round] (dd0) ++(-0.7, .5) .. controls +(-70:1) and +(180:1) .. node[color=black, very near end] [sloped] {\NP} +(3, -3);

        \fill[color=red!10] (dd1) ++(-1, .5) to[out=-70, in=180] ++(4, -3) to ++(0, 3) .. controls +(180:1) and +(110:1) .. cycle;
        \fill[color=red!10,path fading=east] (st2) ++(1,-.5) rectangle ++(1,3);

        \draw[color=red!30,line width=8,line cap=round] (dd1) +(-1, .5) to[out=-70, in=180] node[color=black, near end] [sloped] {\PSPACE} +(3, -2.5);
    \end{pgfonlayer}

\end{tikzpicture}}
    \caption{Straubing-Th{\'e}rien and dot-depth hierarchies: the~\INE~status.}%
    \label{fig:diagram}
\end{figure}

We study the computational complexity of the intersection
non-emptiness for languages from the classes of the
Straubing-Th{\'e}rien~\cite{St81,Th81} and Cohen-Brzozowski's dot-depth
hierarchy~\cite{CoBr71}. Both hierarchies are concatenation hierarchies
that are defined by alternating the use of polynomial and Boolean
closures. Let's be more specific. Let~$\Sigma$ be a finite alphabet. A
language~$L\subseteq\Sigma^*$ is a \emph{marked product} of the
languages $L_0,L_1,\ldots, L_k$, if $L=L_0a_1L_1\cdots a_kL_k$, where
the~$a_i$'s are letters. For a class of languages~$\mathcal M$, the
\emph{polynomial closure} of~$\mathcal M$ is the set of languages that are
finite unions of marked product of languages from~$\mathcal M$.

The concatenation hierarchy of basis~$\mathcal M$ (a class of
languages) is defined as follows (also refer to~\cite{Pin98}): Level~$0$ is~$\mathcal M$, i.e.,
${\mathcal M}_0={\mathcal M}$ and, for each $n\geq 0$,
\begin{enumerate}
\item ${\mathcal M}_{n+1/2}$, that is, level~$n+1/2$, is the polynomial closure of level~$n$ and 
\item ${\mathcal M}_{n+1}$, that is, level~$n+1$, is the Boolean closure of level~$n+1/2$.
\end{enumerate}
The basis of the dot-depth hierarchy is the class of all finite and
co-finite languages\footnote{The dot-depth hierarchy, apart
  level~$\dd{0}$, coincides with the concatenation hierarchy starting
  with the language class
  $\{\emptyset,\{\lambda\},\Sigma^+,\Sigma^*\}$. } and their classes
are referred to as $\dd{n}$ ($\dd{n+1/2}$, \resp), while the basis of
the Straubing-Th{\'e}rien hierarchy is the class of languages that
contains only the empty set and~$\Sigma^*$ and their classes are denoted by
$\st{n}$ ($\st{n+1/2}$, \resp). Their inclusion relation is given by
$$\dd{n+1/2} \subseteq \dd{n+1} \subseteq \dd{n+3/2}\quad\mbox{and}\quad
\st{n+1/2} \subseteq \st{n+1} \subseteq \st{n+3/2},$$ for $n\geq 0$,
and 
$$\st{n-1/2}\subseteq \dd{n-1/2}\subseteq \st{n+1/2}\quad\mbox{and}\quad
\st{n}\subseteq \dd{n}\subseteq \st{n+1},$$ for $n\geq 1$.  In
particular, $\st{0}\subseteq\dd{0}$,\ $\dd{0}\subseteq \dd{1/2}$, and
$\st{0}\subseteq\st{1/2}$.  Both hierarchies are infinite for
alphabets of at least two letters and completely exhaust the class of
star-free languages, which can be described by expressions that use
union, concatenation, and complementation, but \emph{no} Kleene star
operation.  For singleton letter alphabets, both hierarchies collapse
to~$\dd{0}$ and~$\st{1}$, respectively.
Next, we describe the first few levels
of each of these hierarchies:
\begin{description} 
\item[Straubing-Th{\'e}rien hierarchy:] A language of~$\Sigma^*$ is of
  level~$0$ if and only if it is empty or equal to~$\Sigma^*$. 
  The languages of level~$1/2$ are exactly those languages that
  are a finite (possibly empty) union of languages of the form
  $\Sigma^*a_1\Sigma^*a_2\cdots a_k\Sigma^*$, where the~$a_i$'s are
  letters from~$\Sigma$. 
The languages of level~$1$ are
  finite Boolean combinations of languages of the form
  $\Sigma^*a_1\Sigma^*a_2\cdots a_k\Sigma^*$, where the~$a_i$'s are
  letters. 
These languages are also called \emph{piecewise} testable languages.
In particular, all finite and co-finite languages are of level~$1$. 
Finally, the languages of level~$3/2$ of~$\Sigma^*$ are
  the finite unions of languages of the form
  $\Sigma_0^*a_1\Sigma_1^*a_2\cdots a_k\Sigma_k^*$, where the~$a_i$'s
  are letters from~$\Sigma$ and the~$\Sigma_i$ are subsets of~$\Sigma$. 

\item[Dot-depth hierarchy:] A language of~$\Sigma^*$ is of
  dot-depth (level)~$0$ if and only if it is finite or co-finite. 
  The languages of dot-depth~$1/2$ are exactly those languages
  that are a finite union of languages of the form
  $u_0\Sigma^*u_1\Sigma^*u_2\cdots u_{k-1}\Sigma^*u_k$, where~$k\geq
  0$ and the~$u_i$'s are words from~$\Sigma^*$. 
The languages of dot-depth~$1$ are finite Boolean combinations of
  languages of the form $u_0\Sigma^*u_1\Sigma^*u_2\cdots
  u_{k-1}\Sigma^*u_k$, where~$k\geq 0$ and the~$u_i$'s are words
  from~$\Sigma^*$.  
\end{description}
It is worth mentioning that in~\cite{STV01} it was shown that
partially ordered NFAs (with multiple initial states) characterize the
class~$\st{3/2}$, while partially ordered DFAs characterize the class
of $\mathcal R$-trivial languages~\cite{BrFi80}, a class that is
strictly in between $\st{1}$ and $\st{3/2}$. For an automaton~$A$
with input alphabet~$\Sigma$, a state~$q$ is reachable from a
state~$p$, written $p\leq q$, if there is a word $w\in \Sigma^*$ such
that $q\in\delta(p,w)$. An automaton is partially ordered 
if $\leq$ is a partial order.
Partially ordered automata are sometimes also called acyclic or weakly
acyclic automata. We refer to a partially ordered NFA (DFA,
respectively) as poNFA (poDFA, respectively).

%
%

%
The fact that some of our results have a promise looks a bit technical, but the following result implies that we cannot get rid of this condition in general. To this end, we study, for a language class $\mathcal{L}$, the following question of \textsc{$\mathcal{L}$-Membership}.
\begin{itemize}
\item \emph{Input}: A finite automaton $A$.
\item \emph{Question}: Is $L(A)\in\mathcal{L}$?
\end{itemize}
\begin{theorem}
\label{thm:membership-pspace-hardness}
For each level $\mathcal{L}$ of the Straubing-Th\'erien or the dot-depth hierarchies, the \textsc{$\mathcal{L}$-Membership} problem for NFAs is \PSPACE-hard, even when  restricted to binary alphabets.
\end{theorem}
\begin{proof}
For the \PSPACE-hardness, note that each of the classes contains $\{0,1\}^*$ and is closed under quotients, since each class is a positive variety. As \textsc{Non-universality} is \PSPACE-hard for NFAs, we can apply Theorem~3.1.1 of \cite{HunRos78}, first reducing regular expressions to NFAs.
\end{proof}
For some of the lower levels of the hierarchies, we also have containment in $\PSPACE$, but in general, this is unknown, as it connects to the famous open problem if, for instance,  \textsc{$\mathcal{L}$-Membership} is decidable for $\mathcal{L}=\st{3}$; see~\cite{DBLP:journals/tcs/Masopust18,place2017concatenation} for an overview on the decidability status of these questions.
Checking for~$\st{0}$ up to $\st{2}$ and $\dd{0}$ up to~$\dd{1}$ containment for DFAs can be done in \NL\ and is also complete for this class by ideas  similar to the ones used in~\cite{ChHu91}.
%
%


\section{Inside Logspace} 
\label{section:INELogspace}


A language of~$\Sigma^*$ belongs to level~$0$ of the
Straubing-Th{\'e}rien hierarchy if and only if it is empty
or~$\Sigma^*$. The \INE\ problem for language from this language
family is not entirely trivial, because we have to check for
emptiness. Since by our problem definition the property of a language
being a member of level~$0$ is a promise, we can do the emptiness
check within $\AC^0$, since we only have to verify
whether the empty word belongs to the language~$L$ specified by the
automaton. In case $\varepsilon\in L$, then $L=\Sigma^*$; otherwise
$L=\emptyset$. Since in the definition of finite state devices we do
not allow for $\varepsilon$-transitions, we thus only have to check
whether the initial state is also an accepting one.  Therefore, we
obtain:
\begin{theorem}\label{thm:STH-0-AC0}
  The \INE\ problem for DFAs or NFAs accepting languages from~$\st{0}$
belongs to $\AC^0$. 
\end{theorem}
\noindent For the languages of level~$\st{1/2}$
we find the following completeness result.
\begin{theorem}\label{thm:NL-completeness}
  The \INE\ problem for NFAs accepting languages from~$\st{1/2}$ is
  \NL-complete. Moreover, the problem remains \NL-hard even if we
  restrict the input to NFAs over a unary alphabet.
  If the input instance contains only DFAs, the problem
  becomes \L-complete (under weak reductions\footnote{Some form of
    $\AC^0$ reducibility can be employed.}).
\end{theorem}
\noindent Hardness is shown by standard reductions from variants of graph accessibility~\cite{HIM78,Su75b}.
\begin{restatable}{lemma}{RestateLemmaSTHhalfHARD}\label{lem:STH-1/2-hardness}
    The \INE\ problem for NFAs over unary alphabet accepting languages from~$\st{1/2}$ is
  \NL-hard. If the input instance contains only DFAs, the problem
  becomes \L-hard under weak reductions.
\end{restatable}
\begin{proof}
  The \NL-complete graph accessibility problem 2-GAP~\cite{Su75b} is
  defined as follows: given a directed graph $G=(V,E)$ with outdegree
  (at most) two and two vertices~$s$ and~$t$. Is there a path
  linking~$s$ and~$t$ in~$G$?  The problem remains \NL-complete if the
  outdegree of every vertex of~$G$ is exactly two and if the graph is
  ordered, that is, if $(i,j)\in E$, then $i<j$ must be satisfied. The
  complexity of the reachability problem drops to \L-completeness, if
  one considers the restriction that the outdegree is at most one. In this
  case the problem is referred to as 1-GAP~\cite{HIM78}.

  First we consider the \INE\ problem for
  NFAs. The \NL-hardness is seen as follows: let $G=(V,E)$ and
  $s,t\in V$ be an ordered 2-GAP instance. Without loss of generality,
  we assume that $V=\{1,2,\ldots,n\}$, the source vertex $s=1$, and
  the target vertex $t=n$. From~$G$ we construct a unary NFA
  $A=(V,\{a\},\delta,1,n)$, where $\delta(i,a)=\{\,j\mid (i,j)\in
  E\,\}\cup\{i\}$.  The 2-GAP instance has a solution if and only if
  the language accepted by~$A$ is non-empty. Moreover, by construction
  the automaton accepts a language of level~$1/2$, because (i) the
  NFA without $a$-self-loops is acyclic, since~$G$ is ordered, and thus
  does not contain any large cycles and (ii) all states do have
  self-loops.  This proves the hardness and moreover the
  \NL-completeness.

  Finally, we concentrate on the \L-hardness of the
  \INE\ problem for DFAs. Here we use
  the 1-GAP variant to prove our result. Let $G=(V,E)$ and $s,t\in V$
  be a 1-GAP instance, where we can assume that $V=\{1,2,\ldots,
  n\}$,\ $s=1$, and $t=n$. From~$G$ we construct a unary DFA
  $A=(V,\{a\},\delta,1,n)$ with $\delta(i,a)=j$, for $(i,j)\in E$ and
  $1\leq i<n$, and $\delta(n,a)=n$. By construction the DFA~$A$
  accepts either the empty language or a unary language where all
  words are at least of a certain length. In both cases~$L(A)$ is a
  language from level~$1/2$ of the Straubing-Th{\'e}rien
  hierarchy. Moreover, it is easy to see that there is a path in~$G$
  linking~$s$ and~$t$ if and only if $L(A)\neq\emptyset$. Hence, this
  proves \L-hardness and moreover \L-completeness for an intersection
  non-emptiness instance of DFAs.
\end{proof}
%
It remains to show containment in logspace. To this end, we utilize an
alternative characterization of the languages of level~$1/2$ of
the Straubing-Th{\'e}rien hierarchy
 as exactly those languages that are shuffle
ideals. A
language~$L$ is a \emph{shuffle ideal}
if, for every word~$w\in L$ and
$v\in\Sigma^*$, the set $w\shuffle v$ is contained in~$L$, where
$w\shuffle v:=\{\,w_0v_0w_1v_1\ldots w_kv_k\mid
\mbox{$w=w_0w_1\ldots
  w_k$ and $v=v_0v_1\ldots v_k$}
\mbox{ with $w_i,v_i\in\Sigma^*$, for $0\leq i\leq k$}\,\}$.
The operation~$\shuffle$
naturally generalizes to sets. 
%
For the level $\st{1/2}$, we find the following situation.
\begin{restatable}{lemma}{RestateLemmaINTiffEMP}\label{lem:half-level-intersection-nonemptiness-iff-emptiness}
  Let $m\geq 1$ and languages $L_i\subseteq\Sigma^*$, for $1\leq i\leq
  m$, be shuffle ideals, i.e., they belong to~$\st{1/2}$. Then,
  $\bigcap_{i=1}^m L_i\neq\emptyset$ iff the shuffle ideal $L_1L_2\cdots L_m\neq\emptyset$ iff 
$L_i\neq\emptyset$ for every~$i$ with $1\leq i\leq m$. Finally, $L_i\neq\emptyset$, for $1\leq i\leq m$, iff $(a_1a_2\ldots a_k)^{\ell_i}\in L_i$, where $\Sigma=\{a_1,a_2,\ldots a_k\}$ and the shortest word in~$L_i$ is of
  length~$\ell_i$.
\end{restatable}
%
\begin{proof}
  The implication from left to right holds, because if
  $\bigcap_{i=1}^m L_i\neq\emptyset$, then there is a word~$w$ that
  belongs to all~$L_i$, and hence the concatenation $L_1L_2\cdots L_m$
  is nonempty, too. Since this argument has not used the prerequisite
  that the~$L_i$'s belong to the first half level of the
  Straubing-Th{\'e}rien hierarchy, this implication does hold in
  general.

  For the converse implication, recall that a language~$L$ of the first
  half level is a finite (possibly empty) union of languages of the
  form $\Sigma^*a_1\Sigma^*a_2\cdots a_k\Sigma^*$, where the~$a_i$'s
  are letters. Hence, whenever word~$w$ belongs to~$L$, any word of
  the form $uwv$ with $u,v\in\Sigma^*$ is a member of~$L$, too. 
  Now assume that $L_1L_2\cdots L_m\neq\emptyset$, which can be witnessed
  by words $w_i\in L_i$, for $1\leq i\leq m$. But then the word
  $w_1w_2\ldots w_m$ belongs to every~$L_i$, by setting
  $u=w_1w_2\ldots w_{i-1}$ and $v=w_{i+1}w_{i+2}\ldots w_m$ and using
  the argument above. Therefore, the intersection of all~$L_i$, i.e.,
  the set $\bigcap_{i=1}^m L_i$, is nonempty, because of the word
  $w_1w_2\ldots w_m$. 

  The statement that $L_1L_2\cdots L_m$ is an ideal and that
  $L_1L_2\cdots L_m\neq\emptyset$ if and only if~$L_i\neq\emptyset$,
  for every~$i$ with $1\leq i\leq m$, is obvious.

  For the last statement, assume $\Sigma=\{a_1,a_2\ldots,a_k\}$. The
  implication from right to left is immediate, because if $(a_1a_2\ldots
  a_k)^{\ell_i}\in L_i$, for~$\ell_i$ as specified above, then~$L_i$ is
  non-empty. Conversely, if~$L_i$ is non-empty, then there is a
  shortest word~$w$ of length~$\ell_i$ that is contained in~$L_i$. But then
  $(a_1a_2\ldots a_k)^{\ell_i}$ belongs to $w\shuffle\Sigma^*$, which by
  assumption is a subset of the language~$L_i$, since~$L_i$ is an
  ideal. Therefore, $L_i\neq\emptyset$ implies $(a_1a_2\ldots
  a_k)^{\ell_i}\in L_i$, which proves the stated claim.
\end{proof}

\noindent
Now, we are ready to prove containment in logspace.
\begin{lemma}
  \label{lem:ST1/2-L-containment}
  The \INE\ problem for NFAs accepting languages from~$\st{1/2}$ belongs to
  \NL. If the input instance contains only DFAs, the problem is
  solvable in \L.
\end{lemma}
\begin{proof}
  In order to solve the \INE\ problem  
  for given finite automata $A_1,A_2,\ldots, A_m$ with a common input
  alphabet~$\Sigma$, regardless of whether they are deterministic or
  nondeterministic, it suffices to check non-emptiness for all
  languages~$L(A_i)$, for $1\leq i\leq m$, in sequence, because of
  Lemma~\ref{lem:half-level-intersection-nonemptiness-iff-emptiness}. To
  this end, membership of the words $(a_1a_2\ldots a_k)^{\ell_i}$
  in~$L_i$ has to be tested, where~$\ell_i$ is the length of the
  shortest word in~$L_i$. Obviously, all~$\ell_i$ are linearly bounded
  in the number of states of the appropriate finite automaton that
  accepts~$L_i$. Hence, for NFAs as input instance, the test
  can be done on a nondeterministic logspace-bounded Turing machine,
  guessing the computations in the individual NFAs on
  the input word $(a_1a_2\ldots a_k)^{\ell_i}$. For DFAs as input instance,
  nondeterminism is not needed, so that the  procedure can be
  implemented on a deterministic Turing machine. 
\end{proof}

\section{\NP-Completeness}
In contrast to the Straubing-Thérien hierarchy, the \INE problem for languages from the dot-depth hierarchy is already \NP-hard in the lowest level $\dd{0}$. More precisely, \INE for finite languages is \NP-hard~\cite[Theorem~1]{Rampersad2010detecting} and~$\dd{0}$ already contains all finite languages. Hence, the \INE problem for languages from the Straubing-Thérien hierarchy of level $\st{1}$ and above is \NP-hard, too. 
For the levels  $\dd{0}$,\ $\dd{1/2}$,\ $\st{1}$, or~$\st{3/2}$, we give matching complexity upper bounds if the input are DFAs, yielding the first main result of this section proven in \autoref{sec:NP-membership}.
\begin{theorem}
	\label{thm:NPCompleteness}
	The \INE problem for DFAs accepting languages from either $\dd{0}$, $\dd{1/2}$, $\st{1}$, or~$\st{3/2}$ is \NP-complete. The same holds for poNFAs instead of~DFAs. The results hold even for a binary alphabet.
\end{theorem}
For the level $\st{1}$ of the Straubing-Thérien hierarchy, we obtain with the next main theorem a stronger result.
Recall that if all input DFAs accept languages from $\st{1/2}$, the \INE problem is \L-complete due to Lemmata~\ref{lem:STH-1/2-hardness} and~\ref{lem:ST1/2-L-containment}.
\begin{restatable}{theorem}{RestateThmHardnesSTone}
	\label{thm:ST-1-NP-hard}
	The \INE problem for DFAs is \NP-complete even if only one DFA accepts a language from $\st{1}$ and all other DFAs accept languages
	from~$\st{1/2}$ and the alphabet is binary.
\end{restatable}
The proof of this theorem will be given in \autoref{sec:NP-hardness}.

For the level $\dd{0}$, we obtain a complete picture of the complexity of the \INE problem, independent of structural properties of the input finite automata, i.e., we show that here the problem is \NP-complete for general NFAs. 


For the level~$\st{3/2}$, if the input NFA are from the class of poNFA, which characterize level~$\st{3/2}$, then the \INE problem is known to be \NP-complete~\cite{DBLP:journals/corr/abs-1907-13115}. Recall that $\st{3/2}$ contains the levels $\dd{1/2}$, and $\st{1}$ and hence also languages from these classes can be represented by poNFAs.
%
%
%
But if the input automata are given as NFAs without any structural property, then the precise complexity of \INE for $\dd{1/2}$,\ $\st{1}$,\ and $\st{3/2}$ is an open problem and narrowed by \NP-hardness and membership in \PSPACE. We present a ``No-Go-Theorem'' by proving that for an NFA accepting a co-finite language, the smallest equivalent poNFA is exponentially larger in Subsection~\ref{sec:LargePONFA}.

\begin{restatable}{theorem}{RestateThmLargePoNFA}\label{thm:largePONFA}
For every $n \in \mathbb{N}_{\geq 1}$, there exists a language $L_n \in \dd{0}$ on a binary alphabet
	such that $L_n$ is recognized by an NFA of size $O(n^2)$,
	but the minimal poNFA recognizing $L_n$ has more than $2^{n-1}$ states.
\end{restatable}

While for NFAs the precise complexity for \INE of languages from $\st{1}$ remains open, we can tackle this gap by narrowing the considered language class to \emph{commutative} languages in level $\st{1}$; recall that a language $L\subseteq\Sigma^*$ is \emph{commutative} if, for any
$a,b\in\Sigma$ and words $u,v\in\Sigma^*$, we have that $uabv\in L$
implies $ubav\in L$.
We show that for DFAs, this restricted \INE problem  remains \NP-hard, in case the alphabet is unbounded. Concerning membership in \NP, we show that even for NFAs, the \INE problem for \emph{commutative} languages is contained in \NP\ in general and in particular for commutative languages on each level. This generalizes the case of unary NFAs.
Note that for commutative languages, the Straubing-Thérien hierarchy collapses at level $\st{3/2}$. See Subsection~\ref{sec:CommLang} for the proofs.

\begin{restatable}{theorem}{RestateThmComm}
	\label{thm:comm_case}
	The \INE{} problem 
	\begin{itemize}
		\item  is \NP-hard for DFAs accepting \emph{commutative} languages in $\st{1}$, but
		\item  is contained in $\NP$ for NFAs accepting \emph{commutative} languages that might not be star-free.
	\end{itemize}

\end{restatable}

The proof of \NP-hardness for commutative star-free languages in $\st{1}$ requires an arbitrary alphabet. However, we show that \INE is contained in $\XP$ for specific forms of NFAs such as poNFAs or DFAs accepting commutative languages, with the size of the alphabet as the parameter, i.e., for fixed input alphabets, our problem is solvable in polynomial time.


\subsection{\NP-Membership}
\label{sec:NP-membership}
Next, we focus on the \NP-membership part of \autoref{thm:NPCompleteness} and begin by proving that for~$\dd{0}$, regardless of whether the input automata are NFAs or DFAs, the \INE problem is contained in \NP\ and therefore \NP-complete in combination with~\cite{Rampersad2010detecting}.

\begin{lemma}\label{thm:DD-0-NPc}
	The \INE problem for DFAs or NFAs all accepting languages from~$\dd{0}$ is contained in \NP.
\end{lemma}
\begin{proof}
	%
	Let $A_1,A_2,\ldots,A_m$ be 
NFAs accepting languages from $\dd{0}$. If all NFAs accept co-finite languages, which can be verified in deterministic polynomial time, the intersection $\bigcap_{i=1}^m L(A_i)$ is non-empty. Otherwise, there is at least one NFA accepting a finite language, where the longest word is bounded by the number of states of this device. Hence, if $\bigcap_{i=1}^m L(A_i)\neq\emptyset$, there is a word~$w$ of length polynomial in the length of the input that witnesses this fact. Such a~$w$ can be nondeterministically guessed by a Turing machine checking membership of~$w$ in~$L(A_i)$, for all NFAs~$A_i$, in sequence. This shows containment in \NP\ as desired.
\end{proof}

Notice that Masopust and Kr\"otzsch have shown in \cite{DBLP:journals/corr/abs-1907-13115}  that \INE for poDFAs and for poNFAs is \NP-complete. Also the unary case is discussed there, which can be solved in polynomial time. We cannot directly make use of these results, as we consider arbitrary NFAs or DFAs as inputs, only with the promise that they accept languages from a certain level of the studied hierarchies.
In order to prove that for the levels~$\dd{0}$,\ $\dd{1/2}$,\ $\st{1}$, and~$\st{3/2}$, the \INE problem for DFAs is contained in \NP, it is sufficient to prove the claim for $\st{3/2}$ as all other stated levels are contained in~$\st{3/2}$. 
We prove the latter statement by obtaining a bound, polynomial in the size of the largest DFA, on the length of a shortest word accepted by all DFAs\@.
Therefore, we show that for a minimal poNFA $A$, the size of an equivalent DFA is lower-bounded by the size of~$A$ and use a result of
~\cite{DBLP:journals/corr/abs-1907-13115} for poNFAs.
They have shown 
that given poNFAs $A_1,A_2,\ldots,A_m$,
if the intersection of these automata is non-empty, then there exists a word of size at most
$\sum_{i\in\{1,\ldots,m\}}d_i$,
where $d_i$ is the \emph{depth} of $A_i$ \cite[Theorem~3.3]{DBLP:journals/corr/abs-1907-13115}.
Here, the depth of $A_i$ is the length of the longest path (without self-loops) in the state
graph of $A_i$.
%
This result implies that the \INE problem for poNFAs accepting languages from $\st{3/2}$ is contained in \NP. 
We will further use this
result to show that the \INE problem for DFAs accepting languages from $\st{3/2}$ is \NP-complete.
First, we show that the number of states in a minimal poNFA is at most the number of classes in
the Myhill-Nerode equivalence relation.

\begin{restatable}{lemma}{RestateboundSizepoNFA}%
    \label{boundSizepoNFA}
    Let $A = (Q, \Sigma, \delta, q_0, F)$ be a minimal poNFA\@. Then, $L({}_{q_1}A) \neq L({}_{q_2}A)$ for all states
    $q_1, q_2 \in Q$, where ${}_qA$ is defined as $(Q,\Sigma,\delta,q,F)$.
\end{restatable}
\begin{proof}
	Let $A = (Q, \Sigma, \delta, q_0, F)$  be a minimal poNFA and $q_1, q_2 \in Q$
    be two states. Suppose that $L({}_{q_1}A) = L({}_{q_2}A)$. We have two cases.
    \begin{enumerate}
		\item If $q_1$ and $q_2$ are pairwise not reachable from each other, then
		      let $A' = (Q', \Sigma, \delta', q_0, F')$ be the NFA obtained from~$A$, where $q_1$ and~$q_2$
		      are merged into a new state~$q_{1,2}$, so that  $Q' = (Q\setminus\{q_1, q_2\}) \cup \{q_{1,2}\}$,\
		      $\delta'(q_{1,2}, a) = \delta(q_1, a)
              \cup \delta(q_2, a)$, for all $q \in Q'$, $q_{1,2} \in \delta'(q, a)$ if and only if $q_1 \in \delta(q,a)$ or $q_2 \in \delta(q, a)$, and $q_{1,2} \in F'$ if and only if $q_1 \in F$ or
		      $q_2 \in F$. Automata~$A'$ is a partially ordered NFA\@. As $q_1$ and~$q_2$ are not reachable
		      one from the other, they are incomparable in the partial order relation defined by~$A$. Therefore, there is no state $q$ such that $q_1 < q$ and $q < q_2$. One can check
		      that $L(A') = L(A)$, which contradicts the minimality of~$A$.
		\item Otherwise, $q_1$ is reachable from~$q_2$, or $q_2$ is reachable from~$q_1$.
		      Without loss of generality, we assume that $q_2$ is reachable from~$q_1$.
		      Let $A' = (Q', \Sigma, \delta', q_0, F')$ be the NFA obtained from~$A$ in two steps as described next.
		      First, we remove all outgoing transitions from $q_1$ and then we merge $q_1$ and~$q_2$
		      into a new state $q_{1,2}$ as done before. After removing all outgoing
		      transitions from~$q_1$, state~$q_2$ is no longer reachable from $q_1$, therefore,
		      as before, $A'$ is a partially ordered NFA\@. Now we will prove that $L(A) = L(A')$.
		      \begin{itemize}
			      \item Let~$w \in L(A)$. Let $\rho$ be an accepting run in $A$.
			            If~$\rho$ does not contain~$q_1$, then the run obtained by replacing every~$q_2$ by~$q_{1,2}$ is an accepting run in~$A'$.
			            If~$\rho$ contains~$q_1$, then we
			            split $w$ into $w_1$ and $w_2$ such that $w = w_1w_2$ and $w_1$ is the shortest
			            prefix of~$w$ such that, after reading~$w_1$, we reach $q_1$ in $\rho$.
			            Because we merged~$q_1$ and~$q_2$ into~$q_{1,2}$, we have that
                        $q_{1,2} \in \delta'(q_0, w_1)$ in~$A'$. Because $L({}_{q_1}A) = L({}_{q_2}A)$, we
                        have that $L({}_{q_1}A) = L({}_{q_2}A) = L({}_{q_{1,2}}A')$ and therefore
			            $\delta'(q_{1,2}, w_2) \cap F' \neq \emptyset$. So, $w$ is accepted by~$A'$.
			      \item Conversely, let $w \in L(A')$. Let $\rho$ be an accepting run in $A'$.
			            If $\rho$ does not contain $q_{1,2}$, then the same run is accepting in~$A$, too.
			            If $\rho$ contains $q_{1,2}$, we split $w$ into $w_1$ and~$w_2$  such
			            that $w = w_1w_2$, where $w_1$ is the shortest
			            prefix of $w$ such that, after reading~$w_1$, we reach $q_{1,2}$ in~$\rho$.
                        Then, by definition of $q_{1,2}$, $\delta(q_0, w_1) \cap \{q_1, q_2\} \neq \emptyset$,
                        and $\delta(q_1, w_2) \cap F \neq \emptyset$ iff $\delta(q_2, w_2) \cap F \neq
                        \emptyset$ iff $\delta'(q_{1,2}, w_2) \neq \emptyset$.
                        Therefore, $w \in L(A)$.
		      \end{itemize}
              This contradicts the minimality of $A$.
              \qedhere
      \end{enumerate}
\end{proof}
Now, we can use the result from Masopust and Kr\"otzsch to prove that the \INE{} problem
for DFAs accepting languages in $\st{3/2}$ is in \NP\@.
\begin{lemma}%
  \label{lemma:ST-3/2-DFA}
  The \INE{} problem for DFAs accepting
  languages from~$\st{3/2}$ belongs to \NP\@.
\end{lemma}
\begin{proof}
    By \autoref{boundSizepoNFA}, we have that the number of states in a minimal
    poNFA is at most the number of classes of the Myhill-Nerode equivalence relation.
    Hence, given a DFA accepting a language $L\in\st{3/2}$, there exists a smaller
    poNFA that recognizes~$L$. By~
\cite[Theorem~3.3]{DBLP:journals/corr/abs-1907-13115}, if the intersection is not empty, then
    there is a certificate of polynomial size. 
\end{proof}







\subsection{\NP-Hardness}
\label{sec:NP-hardness}
\begin{figure}
	\centering
	\begin{tikzpicture}[scale=0.4,inner sep=1pt,->,>=stealth',shorten >= 1pt,auto,node distance = 1.3cm,semithick, initial text =,state/.style={circle, draw,minimum size=0.7cm, initial text = }] 
		\node[state, initial] (0) {0};
		\node[state, right of=0, node distance = 1.6cm] (1) {\scalebox{0.6}{$i_1$}};
		\node[state, right of=1] (2) {\scalebox{0.6}{$i_1$+$1$}};
		\node[state, right of=2] (3) {\scalebox{0.6}{$i_1$+$2$}};
		\node[state, right of=3, node distance = 1.6cm] (4) {\scalebox{0.6}{$i_2$}};			
		\node[state, right of=4] (5) {\scalebox{0.6}{$i_2$+$1$}};
		\node[state, right of=5] (5') {\scalebox{0.6}{$i_2$+$2$}};
		\node[right of=5, node distance = 1.6cm] (ghost) {};
		\node[state, right of=ghost, node distance = 1.6cm] (6) {\scalebox{0.6}{$n$}};
		\node[state, accepting, right of=6] (7) {\scalebox{0.6}{$n$+$1$}};
		
		\node[state, below of=2] (8) {\scalebox{0.6}{$i_1$+$1$}};
		\node[state, right of=8] (9) {\scalebox{0.6}{$i_1$+$2$}};
		\node[state, below of=5] (10) {\scalebox{0.6}{$i_2$+$1$}};
		\node[state, right of=10] (11) {\scalebox{0.6}{$i_2$+$2$}};			
		\node[state, accepting, below of=6] (12) {\scalebox{0.6}{$n$}};

		\path
		(0) edge[dotted] node[pos=0.4] {$\Sigma$} (1)
		(1) edge node[pos=0.4] {$0$} (2)
		(1) edge node[pos=0.4] {$1$} (8)
		(2) edge node[pos=0.4] {$\Sigma$} (3)
		(3) edge[dotted] node[pos=0.4] {$\Sigma$} (4)
		(4) edge node[pos=0.4] {$0$} (5)
		(4) edge node[pos=0.4] {$1$} (10)
		(5) edge node[pos=0.4] {$\Sigma$} (5')
		(5') edge[dotted] node[pos=0.4] {$\Sigma$} (6)
		(6) edge[pos=0.4] node {$\Sigma$} (7)
		(7) edge[loop right] node {$\Sigma$} (7)
		(8) edge node[pos=0.4, above] {$\Sigma$} (9)
		(9) edge[dotted] node[pos=0.4, above] {$\Sigma$} (10)
		(10) edge node[pos=0.4, above] {$\Sigma$} (11)
		(11) edge[dotted] node[pos=0.4, above] {$\Sigma$} (12)
		(12) edge[loop right] node {$\Sigma$} (12);				
	\end{tikzpicture}
	\caption{DFA $A_{e_i}$ with $L(A_{e_i})=\Sigma^{i_1}\cdot1\cdot\Sigma^{n-i_1-1} \cup
		\Sigma^{i_2}\cdot1\cdot\Sigma^{n-i_2-1}
		\cup
		\Sigma^{\geq n+1}$. A dotted arrow between some states $j$ and $j'$ represents a chain of length $j'-j$ with the same transition labels.}
	\label{fig:st-1-2-lineDFA}
\end{figure}
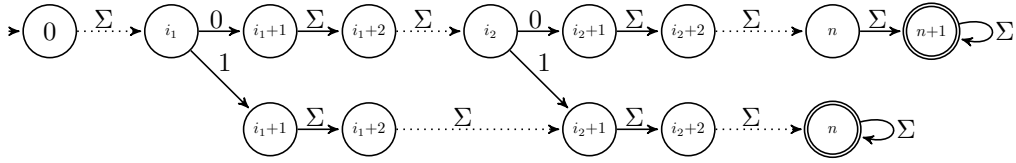
Recall that by~\cite[Theorem~1]{Rampersad2010detecting} \INE for finite languages accepted by DFAs is already \NP-complete. As the level $\dd{0}$ of the dot-depth hierarchy contains all finite language, the \NP-hardness part of Theorem~\ref{thm:NPCompleteness} follows directly from inclusion of language classes. Combining Lemma~\ref{lemma:ST-3/2-DFA}, and \cite[Theorem~3.3]{DBLP:journals/corr/abs-1907-13115} with the inclusion
between levels in the Straubing-Th{\'e}rien and the dot-depth hierarchy, we conclude the proof of \autoref{thm:NPCompleteness}.

\begin{remark}
	Recall that the dot-depth hierarchy, apart form~$\dd{0}$, coincides with the concatenation hierarchy starting with the language class $\{\emptyset,\{\lambda\},\Sigma^+,\Sigma^*\}$. The \INE problem for DFAs or NFAs accepting only languages from $\{\emptyset,\{\lambda\},\Sigma^+,\Sigma^*\}$ belongs to $\AC^0$, by similar arguments as in the proof of Theorem~\ref{thm:STH-0-AC0}.
\end{remark}
We showed in Section \ref{section:INELogspace} that \INE for DFAs, all accepting languages from $\st{1/2}$, belongs to $\L$. If we allow only one DFA to accept a language from $\st{1}$, the problem becomes \NP-hard. The statement also holds if the common alphabet is binary. 
\RestateThmHardnesSTone*
\begin{proof}[Proof idea.]
    The reduction is from \textsc{Vertex Cover}. Let $k\in \mathbb{N}_{\geq 0}$ and let $G=(V, E)$ be a graph with vertex set $V=\{v_0, v_1, \dots, v_{n-1}\}$ and edge set $E=\{e_0, e_1, \dots, e_{m-1}\}$. 
	The only words $w=a_0a_1\dots a_\ell$ accepted by all DFAs will be of length exactly $n=\ell+1$ and encode a vertex cover by: $v_j$ is in the vertex cover if and only if $a_j = 1$. 
	Therefore, we construct for each edge $e_i = \{v_{i_1}, v_{i_2}\} \in E$, with $i_1 < i_2$, a DFA $A_{e_i}$, as depicted in Figure~\ref{fig:st-1-2-lineDFA}, that accepts the language $L(A_{e_i})=\Sigma^{i_1}\cdot1\cdot\Sigma^{n-i_1-1} \cup
	\Sigma^{i_2}\cdot1\cdot\Sigma^{n-i_2-1}
	\cup
	\Sigma^{\geq n+1}$. We show that $L(A_{e_i})$ is from $\st{1/2}$, as it also accepts all words of length at least $n+1$. 
	We further construct a DFA $A_{=n,\leq k}$ that accepts all words of length exactly $n$ that contain at most $k$ letters $1$. 
	The finite language $L(A_{=n, \leq k})$ is the only language from $\st{1}$ in the instance.
\end{proof}
\begin{proof}
The \NP-membership follows from Lemma~\ref{lemma:ST-3/2-DFA} by inclusion of language classes. For the hardness, we give a reduction from the \textsc{Vertex Cover} problem: given an undirected graph $G=(V, E)$ with vertex set $V$ and edge set $E \subseteq V \times V$ and integer $k$. 
	Is there a subset $S \subseteq V$ with $|S| \leq k$ and for all $e \in E$, $S \cap e \neq \emptyset$? 
	If yes, we call $S$ a \emph{vertex cover} of $G$ of size at most $k$. 

    Let $k\in \mathbb{N}_{\geq 0}$ and let $G = (V, E)$ be an undirected graph with vertex set $V = \{v_0, v_1, \dots, v_{n-1}\}$ and edge set $E= \{e_0, e_1, \dots, e_{m-1}\}$.
    From $(G, k)$ we construct $m+1$ DFAs over the common alphabet $\Sigma = \{0, 1\}$. The input word for these automata will encode which vertices are in the vertex cover. Therefore, we assume a linear order on $V$ indicated by the indices of the vertices. More precisely, a word accepted by all automata will have a $1$ at position $j$ if and only if the vertex $v_j$ will be contained in the vertex cover $S$. For a word $w=a_0a_1\dots a_\ell$ with $a_j \in \Sigma$ for $0\leq j \leq \ell$ 
	we denote $w[j]=a_j$. We may call a word $w$ of length $n$ a \emph{vertex cover} and say that the vertex cover covers an edge $e = \{v_{j_1}, v_{j_2}\}$ if $w[j_1] = 1$ or $w[j_2] = 1$.
	
    For every edge $e_i = \{v_{i_1}, v_{i_2}\}$ in~$E$ with $i_1 < i_2$, we construct a DFA $A_{e_i}$ as depicted in~Figure~\ref{fig:st-1-2-lineDFA} consisting of two chains, one of length $n+1$ and one of length $n-(i_1+1)$ (The length of a chain is the number of transitions in the chain). 	
	The DFA is defined as $A_{e_i} = (Q, \Sigma, \delta, q^0, F)$ with state set 
	$Q = \{\,q^j \mid 0 \leq j \leq n+1\,\} \cup \{\,{q'}^j \mid i_1+1 \leq j \leq n\,\}$
    and final states $F = \{\,q^{n+1}, {q'}^n\,\}$.
    We first focus on the states $\{\,q^j \mid 0 \leq j \leq n+1\,\}$. The idea is that there, the first $n+1$ states correspond to the sequence of vertices and reading a $1$ at position $j$ for which $v_j \in e_i$ will cause the automaton to switch to the chain consisting of states $\{\,{q'}^j \mid i_1+1 \leq j \leq n\,\}$. There, only one state is accepting namely the state that we reach after reading a vertex cover of length exactly $n$ that satisfies the edge $e_i$. Note that the paths from $q_0$ to $q'^n$ are one transition shorten than the path from $q_0$ to $q^{n+1}$.
	To be more formal, we define
    $\delta(q^{i_1}, 1) = {q'}^{i_1+1}$ and $\delta(q^{i_2}, 1) = {q'}^{i_2+1}$.
    All other transitions are leading to the next state in the corresponding chain. Formally, we define $\delta(q^{i_1}, 0) = q^{i+1}$ and $\delta({q}^{i_2}, 0)=q^{i_2+1}$, and for all $0 \leq j \leq n$ with $j \notin \{i_1, i_2\}$, we define $\delta(q^j, \sigma) = q^{j+1}$, for both $\sigma \in \Sigma$, and for all $i+1 \leq j \leq n-1$, we define $\delta({q'}^j, \sigma) = {q'}^{j+1}$. We conclude the definition of $\delta$ by defining self-loops for the two accepting states, i.e., we define $\delta(q^{n+1}, \sigma) = q^{n+1}$ and $\delta({q'}^n, \sigma) = {q'}^n$ for both $\sigma\in\Sigma$. Clearly, $A_{e_i}$ is deterministic and of size $\mathcal{O}(n)$.
    
    Note that the only words of length \emph{exactly} $n$ that are accepted by $A_{e_{i}}$ contain a $1$ at position $i_1$ or position $i_2$ and therefore cover the edge $e_i$. All other words accepted by $A_{e_i}$ are of length at least $n+1$. More precisely $A_{e_i}$ accepts \emph{all} words which are of size at least $n+1$.
    Hence, we can describe the language accepted by $A_{e_i}$ as 
    \[L(A_{e_i}) = 
	    \Sigma^{i_1}\cdot1\cdot\Sigma^{n-i_1-1} \cup
	    \Sigma^{i_2}\cdot1\cdot\Sigma^{n-i_2-1}
	    \cup
	    \Sigma^{\geq n+1}.
    \]
    Consider a word $w \in L(A_{e_i})$ of length $n$. W.l.o.g., assume $w[i_1] = 1$. If we insert into $w$ one letter somewhere before or after position $i_1$, then the size of $w$ increases by $1$ and hence $w$ falls into the subset $\Sigma^{\geq n+1}$ of $L(A_{e_i})$. Hence, we can rewrite the language $ L(A_{e_i})$ by the following equivalent expression.
	\[L(A_{e_i}) = 
		\Sigma^{i_1}\cdot \Sigma^*\cdot1\cdot\Sigma^{n-i_1-1}\cdot \Sigma^* \cup
		\Sigma^{i_2}\cdot \Sigma^*\cdot1\cdot\Sigma^{n-i_2-1} \cdot \Sigma^*
		\cup
		\Sigma^{n+1}\Sigma^*.
	\]
    As we can rewrite a language of the form $\Sigma^\ell \Sigma^*$ equivalently as a union of languages of the form $\Sigma^* w_1 \Sigma^* w_2 \dots w_\ell \Sigma^*$ for $w_i \in \Sigma$, for $1 \leq i \leq \ell$, it is clear that $L(A_{e_i})$ is a language of level $\st{1/2}$.


	Next, we define a DFA $A_{= n, \leq k}$ which accepts the finite language of all binary words of length $n$ which contain at most $k$ appearances of the letter $1$. 	
    We define $A_{= n, \leq k} = (\{q_i^j \mid 0 \leq i \leq n+1, 0 \leq j \leq k+1\}, \Sigma, \delta, q_0^0, \{q_n^j \mid j \leq k\})$. The state graph of $A_{= n, \leq k}$ is a $(n,k)$-grid graph, where each letter increases the $x$ dimension represented by the subscript $i$ up to the value $n+1$, and each letter that is a $1$ increases the $y$ dimension represented by the superscript $j$ up to the value $k+1$. More formally, we define $\delta(q_i^j, 0) = q_{i+1}^j$, and $\delta(q_i^j, 1) = q_{i+1}^{j+1}$ for $0 \leq i \leq n$ and $0 \leq j \leq k$; and $\delta(q_{i}^j, \sigma) = q_i^j$ for $i = n+1$ or $j = k+1$. The size of $A_{= n, \leq k}$ is bounded by $\mathcal{O}(nk)$. For readability, we defined $A_{=n, \leq k}$ as a non-minimal DFA. As $L(A_{=n, \leq k})$ is finite, it is of level $\dd{0} \subseteq \st{1}$.
	
	By the arguments discussed above, the set of words accepted by all of the automata $(A_{e_i})_{e_i \in E}$ and $A_{= n, \leq k}$ are of size exactly $n$ and encode a vertex cover for~$G$ of size at most~$k$.
\end{proof}

\newcommand{\AuxLang}{M}

\subsection{Large Partially Ordered NFAs}
\label{sec:LargePONFA}
The results obtained in the last subsection left the precise complexity membership of \INE in the case of input automata being NFAs without any structural properties for the levels $\dd{1/2}$,\ $\st{1}$, and $\st{3/2}$ open. We devote this subsection to the proof of Theorem~\ref{thm:largePONFA}, showing that already for languages of $\dd{0}$ being accepted by an NFA, the size of an equivalent minimal poNFA can be exponential in the size of the NFA.
\RestateThmLargePoNFA*

\begin{proof}
While the statement requires languages over a binary alphabet,
we begin by constructing an auxiliary family $(\AuxLang_n)_{n \in \mathbb{N}_{\geq 1}}$ of languages over an unbounded alphabet.
For all $n \in \mathbb{N}_{\geq 1}$ we then define $L_n$ by encoding $\AuxLang_n$ with a binary alphabet,
and we prove three properties of these languages that directly imply the statement of the Theorem.

For every~$n \in \mathbb{N}_{\geq 1}$, we define the languages~$\AuxLang_n'$ and~$\AuxLang_n''$ over the alphabet~$\{1,2,\ldots,n\}$ as follows.
The language~$\AuxLang_n'$ contains all the words of odd length,
and~$\AuxLang_n''$ contains all the words in which there are two occurrences of some letter~$i \in \{1,2,\ldots,n\}$
with only letters smaller than~$i$ appearing in between.\footnote{The languages $(\AuxLang_n'')_{n \in \mathbb{N}_{\geq 1}}$ were previously studied in~\cite{KleinZ14} with a game-theoretic background. We also refer to~\cite{Nay2011} for similar ``fractal languages.''}
Formally,
\begin{align*}
\AuxLang_n' & = \{\,x \in \{1,2,\ldots,n\}^* \mid |x| \textup{ is odd}\,\},\\
\AuxLang_n'' &= \{\,xiyiz \in \{1,2,\ldots,n\}^* \mid  i \in \{1,2,\ldots,n\},y \in \{1,2,\ldots,i-1\}^*\,\}.
\end{align*}
We then define $\AuxLang_n$ as the union $\AuxLang_n' \cup \AuxLang_n''$.
Moreover, we define $L_n$ by encoding $\AuxLang_n$ with the binary alphabet $\{a,b\}$:
Let us consider the function $\phi_n: \{1,2,\ldots,n\}^* \rightarrow \{a,b\}^*$
defined by $\phi(i_1i_2 \ldots i_m) = a^{i_1}b^{n-i_1}a^{i_2}b^{n-i_2} \ldots a^{i_m}b^{n-i_m}$.
We set $L_n \subseteq \{a,b\}^*$ as the union of $\phi_n(\AuxLang_n)$ with
the language $\{a,b\}^* \setminus \phi(\{1,2,\ldots,n\}^*)$ containing all the words
that are not a proper encoding of some word in $\{1,2,\ldots,n\}^*$.

The statement of the theorem immediately follows from the following claim
{
\begin{restatable}{claim}{RestateClaimPONFA}\label{claim:poNFA}
	\begin{enumerate}
		\item\label{claim:cofinite} The languages $\AuxLang_n$ and $L_n$ are cofinite, thus they are in $\dd{0}$.
		\item\label{claim:smallNFA} The languages $\AuxLang_n$ and $L_n$ are recognized by NFAs of size $n+4$, resp.~$O(n^2)$.
		\item\label{claim:largePONFA} Every poNFA recognizing either $\AuxLang_n$ or $L_n$ has a size greater than $2^{n-1}$.
	\end{enumerate}
\end{restatable}
}
%



\begin{claimproof}[Proof of Item \ref{claim:cofinite}.]
We begin by proving that $\AuxLang_n$ is cofinite.
Note that, by itself, the language~$\AuxLang_n'$
is not in $\dd{0}$, as it is not even star-free.
We show that $\AuxLang_n''$ is cofinite,
which directly implies that $\AuxLang_n = \AuxLang_n' \cup \AuxLang_n''$ is also cofinite.
This follows from the fact that
every word $u \in \{1,2,\ldots,n\}^*$ satisfying $|u| \geq 2^n$ is in $\AuxLang_n''$~\cite{KleinZ14}.
%
%
This is easily proved by induction on~$n$:
If $n=1$, we immediately get that $1^j \in \AuxLang_1''$ for every~$j \geq 2 = 2^1$: such a word contains two adjacent occurrences of $1$.
Now suppose that~$n>1$, and that the property holds for $n-1$.
Every word $u \in \{1,2,\ldots,n\}^*$ satisfying~$|u| \geq 2^n$ can be split into
two parts $u_0$, $u_1$ such that $|u_0|,|u_1| \geq 2^{n-1}$.
We consider two possible cases, and prove that $u \in \AuxLang_n''$ in both of them.
\begin{enumerate}
\item If either $u_0$ or $u_1$ contains no occurrence of the letter~$n$,
then by the induction hypothesis, either $u_0 \in \AuxLang''_{n-1}$ or $u_1 \in \AuxLang''_{n-1}$,
 which directly implies that $u \in \AuxLang_n''$.
\item If both $u_0$ and $u_1$ contain (at least) one occurrence of the letter~$n$,
then $u \in \AuxLang_n''$ since it contains two occurrences of the letter~$n$
with only letters smaller than $n$ appearing in between (the latter part trivially holds, as~$n$ is the largest letter).
\end{enumerate}
Finally, we also get that $L_n$ is cofinite:
for all $u \in \{a,b\}^*$ satisfying $|u| \geq 2^n \cdot n$,
either~$u$ is not a proper encoding of a word of $\{1,2,\ldots,n\}^*$, thus $u \in L_n$,
or $u$ encodes a word $v \in \{1,2,\ldots,n\}^*$ satisfying $|v| \geq 2^n$,
hence $v \in \AuxLang_n$, which again implies that $u \in L_n$.
\end{claimproof}

\begin{claimproof}[Proof of Item \ref{claim:smallNFA}.]
    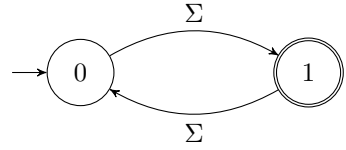
\begin{figure}
        \centering
        \begin{tikzpicture}[>=stealth',->,initial text =]
    \node[draw,circle,state,initial] at (0,0) (0) {$0$};
    \node[draw,circle,state,accepting] at (3,0) (1) {$1$};

    \draw (0) to[bend left] node[above, midway] {$\Sigma$} (1);
    \draw (1) to[bend left] node[below, midway] {$\Sigma$} (0);
\end{tikzpicture}
        \caption{Automaton $A'$ recognizing $\AuxLang_n'$.}%
        \label{fig:m'}
    \end{figure}
    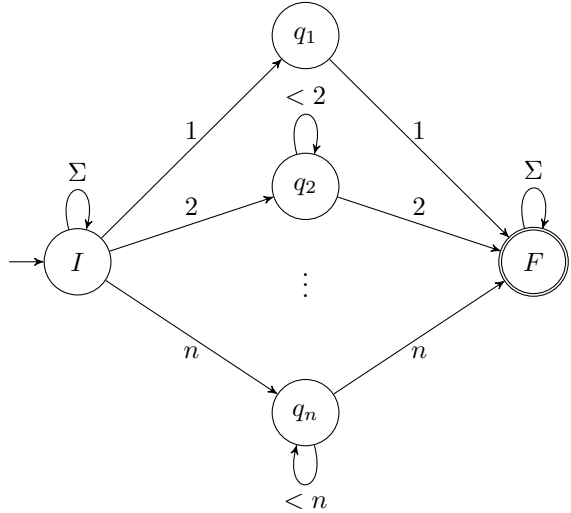
\begin{figure}
        \centering
        \begin{tikzpicture}[>=stealth',->,initial text =]
    \node[draw,circle,state,initial] at (0,0) (I) {$I$};
    \node[draw,circle,state,accepting] at (6,0) (F) {$F$};
    \node[draw,circle,state] at (3,3) (1) {$q_1$};
    \node[draw,circle,state] at (3,1) (2) {$q_2$};
    \node at (3,-0.2) (3) {$\vdots$};
    \node[draw,circle,state] at (3,-2) (n) {$q_n$};

    \draw (I) to[loop above] node[above, midway] {$\Sigma$} (I);
    \draw (F) to[loop above] node[above, midway] {$\Sigma$} (F);
    \draw (2) to[loop above] node[above, midway] {$< 2$} (2);
    \draw (n) to[loop below] node[below, midway] {$< n$} (n);
    \draw (I) to node[above, midway] {$1$} (1);
    \draw (1) to node[above, midway] {$1$} (F);
    \draw (I) to node[above, midway] {$2$} (2);
    \draw (2) to node[above, midway] {$2$} (F);
    \draw (I) to node[below, midway] {$n$} (n);
    \draw (n) to node[below, midway] {$n$} (F);
\end{tikzpicture}
        \caption{Automaton $A''$ recognizing $\AuxLang_n''$.}%
        \label{fig:m''}
    \end{figure}
We first construct an NFA $A$ of size $n+4$ recognizing $\AuxLang_n = \AuxLang_n'\cup \AuxLang_n''$
as the disjoint union of an NFA $A'$ (\autoref{fig:m'}) of size $2$ recognizing $\AuxLang_n'$
and an NFA $A''$ (\autoref{fig:m''}) of size $n+2$ recognizing~$\AuxLang_n''$.
The language $\AuxLang_n'$ of words of odd length is trivially recognized by an NFA of size $2$,
thus we only need to build an NFA~$A'' = (Q, \{1,2,\ldots,n\},\delta, q_I, \{q_F\})$
of size~$n+2$ that recognizes~$\AuxLang_n''$.
The state space~$Q$ is composed of the start state~$q_I$,
the single final state~$q_F$,
and~$n$ intermediate states~$\{q_1,q_2, \ldots, q_n\}$.
The NFA $A''$ behaves in three phases:
\begin{enumerate}
	\item First, $A''$ loops over its start state
	until it non-deterministically guesses that it will read two copies of some $i \in \Sigma$ with smaller letters in between:
	$\delta(q_I,i) = \{q_I,q_i\}$ for all $i \in \Sigma$.
	\item To check its guess, $A''$ loops in $q_i$ while reading letters smaller than $i$ until it reads a second $i$:
	$\delta(q_i,j) = \{q_i\}$ for all $j \in \{1,2,\ldots,i-1\}$ and $\delta(q_i,i) = \{q_F\}$. 
	\item The final state $q_F$ is an accepting sink: $\delta(q_F,j) = \{q_F\}$ for all $j \in \Sigma$.
\end{enumerate}
This definition guarantees that $A''$ accepts the language $\AuxLang_n''$.

Finally, we build an NFA $B$ of size $O(n^2)$
that recognizes $L_n$ by following similar ideas.
Once again, $B$ is defined as the disjoint union of two NFAs $B'$ and $B''$:
The NFA $B'$ uses $4n$ states to check that either the input is
\emph{not} a proper encoding,
or the input encodes a word $u \in \{1,2,\ldots,n\}^*$ of odd length.
Then, the NFA $B''$ with $O(n^2)$ states is obtained by adapting the NFA $A''$ to the encoding of the letters $\{1,2,\ldots,n\}$:
we split each of the $2n$ intermediate transitions of $A''$ into $n$ parts by adding $n-1$ states,
and we add $2(n-1)$ states to each self-loop of $A''$ in order to check that the encoding of an adequate letter is read.
\end{claimproof}

\begin{claimproof}[Proof of Item \ref{claim:largePONFA}.]
It is sufficient to prove the result for $\AuxLang_n$,
as we can transform each poNFA $A = (Q, \{a,b\},\delta_A, q_I, F)$ recognizing $L_n$
into a poNFA $B = (Q, \{1,2,\ldots,n\},\delta_B, q_I, F)$ recognizing~$\AuxLang_n$ with the same set of states
by setting $\delta_B(q,i) = \delta_A(q,a^ib^{n-i})$. 

Note that, by itself, the language $\AuxLang_n''$ is recognized
by the poNFA $A$ of size $n+2$ defined in the proof of Item~\ref{claim:smallNFA}.
Let $A'$ be a poNFA recognizing $\AuxLang_n$.
To show that $A'$ has more than~$2^{n-1}$ states,
we study its behavior on the \emph{Zimin words},
defined as follows:
\[
\textup{Let $u_1 = 1$ and $u_j = u_{j-1}ju_{j-1}$  for all $1 < j \leq n$.}
\]
For instance, $u_4 = 121312141213121$.
It is known that
$|u_j| = 2^j-1$ and
$u_j \notin \AuxLang_n''$
for every $1 \leq j \leq n$~\cite{KleinZ14}.
These two properties are easily proved by induction on~$j$:
Trivially, $u_1$ is not in $\AuxLang_1''$ and its size is $1 = 2^1-1$.
Now suppose that $j>1$ and that $u_{j-1}$ satisfies both properties:
$|u_{j-1}| = 2^{j-1}-1$ and
$u_{j-1} \notin \AuxLang_n''$.
The first property follows immediately from the induction hypothesis.
\[
|u_j| = |u_{j-1}ju_{j-1}| = 2 \cdot |u_{j-1}| + 1 = 2 \cdot (2^{j-1}-1) + 1 = 2^j-1;\\
\]
To prove the induction step for the second property,
we suppose, towards building a contradiction,
that $u_j \in \AuxLang_n''$.
Then $u_j$ contains two occurrences of some letter $i \in \{1,2,\ldots,n\}$
with only letters smaller than $i$ appearing in between.
Since~$u_j$ contains only one occurrence of the letter~$j$ and no letter is greater than $j$,
$i$ is strictly smaller than $j$.
Moreover, as only letters smaller than $i$ (thus no $j$) can appear between these two occurrences,
they both need to appear in one of the copies of $u_{j-1}$.
Therefore $u_{j-1}$ is also in~$\AuxLang_n''$, which contradicts the induction hypothesis.

To conclude, remark that the word $u_n$ is not in $\AuxLang_n''$,
but since $|u_n| = 2^n-1$ is odd, it is in~$\AuxLang_n = L(A')$.
Consider a sequence $\rho \in Q^*$ of states leading $A'$ from its start state to a final state over the input $u_n$.
Observe that the word $u_n$ contains $2^{n-1}$ occurrences of the letter~$1$,
and deleting (any) one of these occurrences results in a word of even length that is still not in $\AuxLang_n''$,
thus it is also not in $\AuxLang_n= L(A')$.
This proves that the sequence $\rho$ cannot loop over any of the $1$'s in $u_n$.
Moreover, as $A'$ is partially ordered by assumption, once it leaves a state, it can never return to it.
Therefore, $\rho$ contains at least $2^{n-1}+1$ distinct states while processing the $2^{n-1}$ occurrences of $1$ in $u_n$,
which shows that the automaton $A'$ has more than $2^{n-1}$ many states.
\end{claimproof}
This concludes the proof.
\end{proof}


\subsection{Commutative Star-Free Languages}\label{sec:CommLang}


In the case of commutative languages, we have a complete picture of
the complexities for both hierarchies, even for arbitrary input NFAs.
Observe, that commutative languages generalize unary
languages, where it is known that for unary star-free languages both hierarchies collapse. For commutative star-free languages, a similar result holds, employing~\cite[Prop.~30]{Hoffmann2021ciaa}.

\begin{restatable}{theorem}{RestateTheoremCollapse}\label{thm:collapse_commutative}
  For \emph{commutative star-free languages} the levels $\st{n}$ of
  the Straubing-Th\'erien and $\dd{n}$ of the dot-depth hierarchy
  coincide for all full and half levels, except for $\st{0}$ and
  $\dd{0}$. Moreover, the hierarchy collapses at level one.
\end{restatable}
\begin{proof}
  The strict inclusion $\st{0}\subset\dd{0}$ even in the commutative
  case is obvious.  Since $\st{1/2} \subseteq \dd{1/2}$ we only need
  to show the converse inclusion in the case of commutative languages.
  For the sake of notational simplicity, we shall give the proof only
  in a special case. Observe that, by commutativity, if $\Sigma^* ab
  \Sigma^* \subseteq L$, then $\Sigma^* a \Sigma^* b \Sigma^*
  \subseteq L$; moreover, $\Sigma^* ab \Sigma^* \subseteq \Sigma^* a
  \Sigma^* b \Sigma^*$.  Using this idea repeatedly for marked
  products as they describe languages from $\dd{1/2}$, we can write
  them as equivalent polynomials used for defining languages
  from~$\st{1/2}$.
%
%

  It remains to show that every commutative star-free language is
  contained in~$\st{1}$. As shown
  in~\cite[Prop.~30]{Hoffmann2021ciaa}, every star-free commutative
  language can be written as a finite union of languages of the form
  $L = \comm(u) \shuffle \Gamma^*$ for some $u \in \Sigma^*$ and
  $\Gamma \subseteq \Sigma$. Here
  $\comm(u)=\{\,w\in\Sigma^*\mid\mbox{$|u|_a=|w|_a$ for every
    $a\in\Sigma$}\,\}$, where $|w|_a$ is equal to the number of
  occurrences of~$a$ in~$w$. Since $\comm(u)$ is a finite language,
  clearly, language~$L$ is equal to the finite union of all
  $v\shuffle\Gamma^*$ for $v\in\comm(u)$, and thus belongs to
  $\st{3/2}$, since $\Gamma\subseteq\Sigma$.

  Now, note that $v \shuffle \Sigma^* = \Sigma^* v_1 \Sigma^* \cdots \Sigma^* v_{|v|} \Sigma^*$, where $v = v_1 \cdots v_{|u|}$ with $v_i \in \Sigma$,
  is in level one of the hierarchy. Further,
  \[
   v \shuffle \Gamma^* = ( v \shuffle \Sigma^* ) \cap \overline{ \bigcup_{a \in \Sigma \setminus \Gamma} \comm(va) \shuffle \Sigma^* }.
  \]
  Hence, we can conclude containment in $\st{1}$.
\end{proof}

%
Next we will give the results, summarized in \autoref{thm:comm_case},
for the case of the commutative (star-free) languages. The \NP-hardness follows by a reduction from \textsc{3-CNF-SAT}.

%

\begin{restatable}{lemma}{RestateLemmaCommHard}
\label{lem:com_np_hardness}
	The \INE problem is
	\NP-hard for DFAs accepting \emph{commutative} languages in $\st{1}$.
\end{restatable}
\begin{proof}
  The \NP-complete \textsc{3-CNF-SAT} problem is defined as follows:
  given a Boolean formula~$\varphi$ as a set of clauses
  $C=\{c_1,c_2,\ldots, c_m\}$ over a set of variables
  $V=\{x_1,x_2,\ldots,x_n\}$ such that $|c_i|\leq 3$ for $i\leq m$. Is
  there a variable assignment $\beta:V\to\{0,1\}$ such that~$\varphi$
  evaluates to true under~$\beta$?

  Let~$\varphi$ be a Boolean formula in \textsc{3-CNF} with clause set
  $C = \{c_1,c_2, \ldots, c_m\}$ and variable set $V = \{x_1,x_2,
  \ldots, x_n\}$.  Let $\Sigma = \{ x_1,x_2, \ldots, x_n,
  \overline{x}_1,\overline{x}_2, \ldots, \overline{x}_n \}$. It is
  straightforward to construct polynomial-size DFAs for the following
  languages from~$\dd{1}$:
 \begin{align*} 
   L_{c_i} = \bigcup_{x \in c_i} \Sigma^* x
   \Sigma^*\quad\mbox{and}\quad L_{x_j} = \Sigma^*\setminus
   (\Sigma^* x_j \Sigma^* \overline{x}_j \Sigma^* \cup \Sigma^*
   \overline{x}_j \Sigma^* x_j \Sigma^*)\,,
 \end{align*}
 where $1\leq i \leq m$ and $1\leq j\leq n$. Then, the intersection
 of all $L_{c_i}$ and all $L_{x_j}$ is non-empty if and only if the
 \textsc{3-CNF-SAT} instance~$\varphi$ is satisfiable.
\end{proof}
The upper bound shown next also holds for arbitrary commutative languages.

\begin{theorem}
\label{lem:com_containment_np}
The \INE problem for NFAs accepting arbitrary, i.e., not necessarily
star-free, \emph{commutative} languages is in $\NP$.
\end{theorem}
\begin{proof}
  It was shown in~\cite{StMe73} that \INE is $\NP$-complete for unary
  NFAs as input. Fix some order $\Sigma = \{a_1, a_2,\ldots, a_r\}$ of
  the input alphabet.  Let $A_1, A_2, \ldots, A_m$ be the NFAs
  accepting commutative languages with $ A_i = (Q_i, \Sigma, \delta_i,
  q_{0,i}, F_i)$ for $1\leq i\leq m$. Without loss of generality, we
  may assume that every~$F_i$ is a singleton set, namely
  $F_i=\{q_{f,i}\}$. For each $1\leq i\leq m$ and $1\leq j\leq r$, let
  $ B_{i,j}$ be the automaton over the unary alphabet~$\{a_j\}$
  obtained from $A_i$ by deleting all transitions labeled with letters different
  from~$a_j$ and only retaining those labeled with~$a_j$. Each $B_{i,j}$
  will have one initial and one final state.
  Let $\vec{q}_{0}=(q_{0,1}, q_{0,2},\ldots, q_{0,m})$ be the tuple of initial
  states of the NFAs; they are the initial states of $B_{1,1},
  B_{2,1}, \ldots, B_{m,1}$, respectively.  Then, nondeterministically
  guess further  tuples~$\vec{q}_j$ from $Q_1 \times
  Q_2\times \ldots \times Q_m$ for $1\leq j\leq r-1$.  The $j$th tuple
  is considered as collecting the final states of the~$B_{i,j}$ but
  also as the start states for the $ B_{i,j+1}$. Finally, let
  $\vec{q_f}=(q_{f,1}, q_{f,2},\ldots, q_{f,m})$ and consider this as
  the final states of $ B_{1,r}, B_{2,r},\ldots, B_{m,r}$.  Then, for
  each $1\leq j\leq r$ solve \INE\ for the unary automata $B_{1,j},
  B_{2,j},\ldots, B_{m,j}$.  If there exist words~$w_j$ in the
  intersection of $L(B_{1,j}), L(B_{2,j}),\ldots, L(B_{m,j})$, for
  each $1\leq j\leq r$, then, by commutativity, there exists one in
  $a_1^* a_2^*\cdots a_r^*$, namely, $w_1w_2\cdots w_m$, and so the
  above procedure finds it. Conversely, if the above procedure finds a
  word, this is contained in the intersection of the languages
  induced by the $A_i$'s.
\end{proof}


For fixed
alphabets, we have a polynomial-time algorithm, showing that the
problem is in \XP\ for alphabet size as a parameter, for a class of
NFAs generalizing, among others, poNFAs and DFAs (accepting star-free
languages). This is in contrast to the other results on the
\INE\ problem in this paper. We say that an
NFA $A=(Q,\Sigma,\delta,q_0,F)$ is \emph{totally star-free}, if 
the language accepted by $_q A_{p}= (Q, \Sigma, \delta,q, \{p\})$ is star-free
for
any states $q,p\in Q$. For instance, partially ordered NFAs
are totally star-free. 

An example of a non-totally star-free NFA accepting a star-free language is given next. 
Consider the following NFA
$A=(\{q_0,q_1,q_2,q_3\},\delta,q_0,\{q_0,q_2\})$ with
$\delta(q_0,a) =\{q_1,q_2\}$,\
$\delta(q_1,a) =\{q_0\}$,\
$\delta(q_2,a) =\{q_3\}$,
and $\delta(q_3,a) =\{q_2\}$
that accepts the language~$\{a\}^*$. The automaton is depicted in Figure~\ref{fig:non-totally-NFA}.
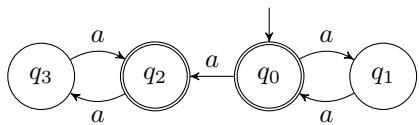
\begin{figure}
  \centering
  \begin{tikzpicture}[>=stealth',->,initial text =]
    \node[draw,circle,state] at (0,0) (3) {$q_3$};
    \node[draw,circle,state,accepting] at (1.5,0) (2) {$q_2$};
    \node[draw,circle,state,initial above,accepting] at (3,0) (0) {$q_0$};
    \node[draw,circle,state] at (4.5,0) (1) {$q_1$};

    \draw (0) to[bend left] node[above, midway] {$a$} (1);
    \draw (1) to[bend left] node[below, midway] {$a$} (0);
    \draw (2) to[bend left] node[below, midway] {$a$} (3);
    \draw (3) to[bend left] node[above, midway] {$a$} (2);
    \draw (0) to node[above, midway] {$a$} (2);
\end{tikzpicture}
  \caption{An example of a non-totally star-free NFA that accepts a star-free language.}
  \label{fig:non-totally-NFA}
\end{figure}
Yet, neither
$L(_{q_0}A_{q_0})=\{aa\}^*$ nor $L(_{q_0}A_{q_2})=\{a\}\{aa\}^*\cup\{\varepsilon\}$
are star-free.

\noindent The proof of the following theorem uses classical results of Chrobak and Schützenberger~\cite{DBLP:journals/tcs/Chrobak86,DBLP:journals/iandc/Schutzenberger65a}.
\nocite{DBLP:conf/wia/Gawrychowski11}

\begin{restatable}{theorem}{RestateTheoremXP}\label{thm:XP_case}
  The \INE{} problem for totally star-free NFAs accepting
  \emph{star-free commutative} languages, i.e., commutative languages
  in $\st{3/2}$, is contained in \XP\ (with the size of the
  alphabet as the parameter).
\end{restatable}


The proof of Theorem~\ref{thm:XP_case}
is based on a combinatorial number theoretical result that might be of
independent interest.

\begin{lemma}%
\label{lem:threshold_arithmetic_progression_stuff}
Let $n\geq 1$ and $t_i, p_i \in \mathbb{N}_{\geq 0}$ for $1\leq i\leq n$. Set
$X= \bigcup_{i=1}^n \left(t_i + \mathbb{N}_{\geq 0} \cdot p_i\right)$,
where $\mathbb{N}_{\geq 0} \cdot p_i = \{\, x \cdot p_i \mid x \in \mathbb{N}_{\geq 0} \}$.
 If there exists a threshold $T \ge 0$ such that
 $ \{\, x \in \mathbb{N}_{\geq 0} \mid x \ge T\, \}
   \subseteq X$,
 then already for $T_{\max}=\max\{\, t_i \mid 1\leq i\leq n\,\}$, we find 
 $
 \{\, x \in \mathbb{N}_{\geq 0} \mid x \ge  T_{\max}\,\}
   \subseteq X
 $.
\end{lemma}
\begin{proof}
 The assumption basically says that every integer~$y$ greater than~$T-1$
 is congruent to~$t_\ell$ modulo~$p_\ell$ for some $1\leq\ell\leq n$.
 More specifically, if $x$ is an arbitrary number with $x \ge T_{\max}$, 
 then $y=x + T \cdot \lcm\{p_1, p_2,\ldots, p_n\}$ is congruent to~$t_\ell$
 modulo~$p_\ell$ for some $1\leq\ell\leq n$.
 But this implies that~$x$ itself is congruent to~$t_\ell$ modulo~$p_\ell$,
 and so, as $x \ge t_\ell$, we can write $x= t_\ell + k_\ell \cdot p_\ell$
 for some $k_\ell \ge 0$, i.e., $x \in X$.
\end{proof}

This number-theoretic result can be used to prove a polynomial bound for star-free unary languages on an equivalence resembling Schützenberger's characterization of star-freeness~\cite{DBLP:journals/iandc/Schutzenberger65a}.

\begin{lemma}
\label{lem:unary_polybound}
Let $L$ be a unary star-free language specified by an NFA $A$ with $n$ states.
Then, there is a number $N$ of order $\mathcal{O}(n^2)$ such that $a^N\in L$ if and only if
for all $k\in \mathbb{N}_{\geq 0}$, $a^{N+k}\in L$.
\end{lemma}

\begin{proof}
  By a classical result of Chrobak~\cite{DBLP:journals/tcs/Chrobak86},
  the given NFA~$A$ on~$n$ states can be transformed  into a normal
  form where we have an initial tail with length at most $\mathcal{O}(n^2)$
  that branches at a common endpoint into several cycles, where
  every cycle is of size at most $n$, see
  \cite[Lemma~4.3]{DBLP:journals/tcs/Chrobak86}.  Moreover, this
  transformation can be performed in polynomial
  time~\cite{DBLP:conf/wia/Gawrychowski11}.  Note that a unary
  star-free language is either finite or co-finite~\cite{Brz76}. If~$L$
  is finite, then there are no final states on the cycles and we can
  set~$N$ to be equal to the length of the tail, plus one.  Otherwise,
  if $L\subseteq \{a\}^*$ is co-finite, then it can be expressed as a union
  of a finite language corresponding to the final states on the tail and
  finitely many languages of the form $\{\,a^{\ell}\mid \ell\in
  (t+\mathbb{N}_{\geq 0}\cdot p)\,\}$, where the numbers~$t$ and~$p$ are induced by the Chrobak normal form. Then 
  we can apply Lemma~\ref{lem:threshold_arithmetic_progression_stuff},
  where the set $X$ is built from the $t$'s and~$p$'s, and where the~$t$'s
  are bounded by~$T_{\max}$, the sum of the longest tail and the
  largest cycle, plus one. Note that $T_{\max}$ is in~$\mathcal{O}(n^2)$ and that the threshold from Lemma~\ref{lem:threshold_arithmetic_progression_stuff} guarantees that
  every word~$a^\ell$ with $\ell\geq T_{\max}$ is a member of~$L$, as desired.
\end{proof}


\RestateTheoremXP*

\begin{proof}
 Let $A_i = (Q_i, \Sigma, \delta_i, q_i, F_i)$, for $i \in \{1,2,\ldots,m\}$, be
  totally star-free NFAs accepting commutative languages. Let $n_i=|Q_i|$ be the number of states of $A_i$. 
 Fix some order $\Sigma = \{a_1, a_2,\ldots, a_r\}$.
  
 For $1\leq i\leq m$ and $1\leq j\leq r$, as well as $q,p\in Q_i$, 
 let the automaton $B_{i,j,q,p} = (Q_i, \{a_j\}, \delta_i, q, \{p\})$
 be  obtained from~$A_i$ by deleting all transitions
 not labeled with the letter $a_j$ and only retaining those labeled with $a_j$.
Further, let $A_{i,q,p}$ be  obtained from $A_i$ by taking $q$ as (new) initial state and $p$ as the new (and only) final state.
As $A_i$ is totally star-free, $L(A_{i,q,p})$ is also star-free. By Schützenberger's Theorem characterizing star-freeness~\cite{DBLP:journals/iandc/Schutzenberger65a}, 
it is immediate that 
$L(A_{i,q,p})\cap\Gamma^*$ is also star-free
for each $\Gamma\subseteq\Sigma$.
In particular, $L(B_{i,j,q,p})=L(A_{i,q,p})\cap\{a_j\}^*$ is star-free and commutative.

%
Recall that 
  $\comm(u)=\{\,w\in\Sigma^*\mid\mbox{$|u|_a=|w|_a$ for every
    $a\in\Sigma$}\,\}$, where $|w|_a$ is equal to the number of letters~$a$ in~$w$. Moreover, $\comm(L)=\bigcup_{v\in L}\comm(v)$.
%
 By commutativity, the following property is clear:
$$L(A_i)=\comm\left(\bigcup_{p_1,p_2,\dots,p_{r-1}\in Q_i}\bigcup_{p_r\in F_i}L(B_{i,1,q_i,p_1})\cdot L(B_{i,2,p_1,p_2})\cdots L(B_{i,r,p_{r-1},p_r})\right).$$

As $A_i$ accepts a commutative language, by ordering the letters, we find that $w\in L(A_i)$ if and only if $a_1^{\ell_1}a_2^{\ell_2} \cdots a_r^{\ell_r} 
  \in L( A_i)$, for $\ell_j$ being the number of occurrences of $a_j$ in $w$, with $1\leq j\leq r$.
Furthermore, the 
word  $a_1^{\ell_1}a_2^{\ell_2} \cdots a_r^{\ell_r}$ is in $L(A_i)$ if and only if for all~$j$ with $1\leq j\leq r$, there is a state $p_j\in Q_i$ such that  
$a_j^{\ell_j} \in L( B_{i,j,p_{j-1},p_j})$,
 where $p_0=q_i$ and $p_r\in F_i$. 
We can apply Lemma~\ref{lem:unary_polybound} to get constants $N_{i,1},N_{i,2},\dots, N_{i,r}\in \mathcal{O}(n_i^2)$ such that checking membership of $a_j^{\ell_j}$ in $L( B_{i,j,p_{j-1},p_j})$ can be restricted to checking membership for a word of length at most~$N_j$.  
  Now, we describe a polynomial-time
 procedure to solve \INE
 for fixed alphabets.
 Set $N_i = \max\{N_{i,1},N_{i,2}, \ldots, N_{i,r}\}$
 with the numbers $N_{i,j}$ from above.
 Then, we know that a word $a_1^{\ell_1}a_2^{\ell_2} \cdots a_r^{\ell_r}$
 is accepted by an input automaton~$A_i$ if and only if
 the word $a_1^{\min\{\ell_1, N_i\}}a_2^{\min\{\ell_2, N_i\}} \cdots a_r^{\min\{\ell_r,N_i\}}$
 is accepted by it. If we let $N=\max\{N_1,N_2,\ldots,N_r\}$ and $n=\max\{n_1,n_2,\ldots,n_m\}$,
 we only need to test the $(N+1)^r$
 many words $a_1^{i_1}a_2^{i_2} \cdots a_r^{i_r}$
 with $0\leq i_j \leq N$ and $1\leq j\leq r$ 
 if we can find a word among them that is accepted by all automata $A_i$ for $1\leq i\leq m$.
Altogether, ignoring polynomial factors, this leads to a running time of the form $(N+1)^r$. 
\end{proof}

\begin{remark}
 Note that Theorem~\ref{thm:XP_case}
 does not hold for arbitrary commutative languages concerning a fixed alphabet, but only for star-free commutative languages, since
 in the general case, the problem is \NP-complete
 even for languages over a common unary alphabet~\cite{StMe73}.
\end{remark}


\begin{figure}
\end{figure}
\section{\PSPACE-Completeness}

\newcommand\VS[1][.3em]{%
	\mbox{\kern.1em\vrule height.3ex}%
	\vbox{\hrule width#1}%
	\hbox{\vrule height.3ex}
}

%

Here, we prove that even when restricted to languages from $\dd{1}$ or $\st{2}$, \INE\ is \PSPACE-complete, as it is for unrestricted DFAs or NFAs.
We will profit from the close relations of \INE\ to  the \textsc{Non-universality} problem for NFAs: Given an NFA~$A$ with input alphabet~$\Sigma$, decide if $L(A)\neq\Sigma^*$.
Conversely, we can also observe that \textsc{Non-universality}  for NFAs is \PSPACE-complete for languages from $\dd{1}$.

\begin{theorem}
\label{theorem:PSPACE-Completeness}
The \INE\ problem for DFAs or NFAs accepting languages from
$\dd{1}$ or $\st{2}$ is \PSPACE-complete, even for binary input alphabets.
\end{theorem}
As $\dd{1} \subseteq \st{2}$,
it is sufficient to show that the problem is \PSPACE-hard for $\dd{1}$.
While without paying attention to the size of the input alphabet, this result can be readily obtained by re-analyzing Kozen's original proof in \cite{Kozen1977Lower},
the restriction to binary input alphabets needs some more care. 
We modify  the proof of Theorem~3 in \cite{KMT17} that showed  \PSPACE-completeness for \textsc{Non-universality} for poNFAs (that characterize the level $3/2$ of the Straubing-Th\'erien hierarchy).
Also, it can be observed that the languages involved in the intersection are actually locally testable languages.
Without giving details of definitions,  we can therefore formulate:

 \begin{corollary}The \INE\ problem for DFAs or NFAs accepting locally testable languages is \PSPACE-complete, even for binary input alphabets.
\end{corollary}
%
\begin{proof}
To see our claims, we re-analyze the proof of Theorem~3 in \cite{KMT17} that shows \PSPACE-completeness for the closely related  \textsc{Non-universality} problem for NFAs.
Similar to Kozen's original proof, this gives a reduction from the general word problem of deterministic polynomial-space bounded Turing Machines. 
In the proof of Theorem~3 in \cite{KMT17} that showed  \PSPACE-completeness for \textsc{Non-universality} for poNFAs (that characterize the level $3/2$ of the Straubing-Th\'erien hierarchy),  a polynomial number of binary languages~$L_i$ was constructed such that $\bigcup_i L_i\neq\{0,1\}^*$ if and only if the $p$-space-bounded Turing machine~$M$, where~$p$ is some polynomial, accepts a word $x\in\{0,1\}^*$ using space $p(|x|)$. Observe that each of the languages $L_i$
is a polynomial union of languages of the forms $E\{0,1\}^*$,\ $\{0,1\}^*E$,  $\{0,1\}^*E\{0,1\}^*$, or $E$ for finite binary languages~$E$.
This means that each $L_i$ belongs to~$\dd{1/2}$.
Now, observe that $\bigcup_iL_i\neq\{0,1\}^*$ if and only if $\bigcap_i\overline{L_i}\neq\emptyset$. As $\overline{L_i}\in \dd{1}$ and each $L_i$ (and hence its complement $\overline{L_i}$) can be described by a polynomial size DFA, the claims follow.
\end{proof}

\noindent By the proof of Theorem~3 in \cite{KMT17}, also $\bigcup_i L_i$ belongs to $\dd{1}$, so that we can conclude:

\begin{corollary}
The \textsc{Non-universality} problem for NFAs accepting languages from
$\dd{1}$ is \PSPACE-complete, even for binary input alphabets.
\end{corollary}

We now present all proof details, because the construction is somewhat subtle.

The proof is based on simulating a $p$-space-bounded Turing machine $M$. We are interested in simulating a run of $M$ on a string $x$.
Its configurations are encoded as words over an alphabet~$\Delta$, so that with the help of the enhanced alphabet $\Delta_\#=\Delta\cup\{\#\}$, runs of $M$ can be encoded,  with $\#$ serving as a separator between configurations. More precisely, if $\Sigma_M$ is the input alphabet of $M$, $\Gamma_M$ (containing a special \emph{blank} symbol~$\VS$) is the tape alphabet, and $Q_M$ is the state alphabet, then transitions take the form $f_M:Q_M \times \Gamma_M \to Q_M \times \Gamma_M \times \{L,R\}$, where $L,R$ indicate the movements of the head.
For simplicity, define $\Delta=\Gamma_M\times (Q_M\cup\{\$\})$. A configuration $\gamma\in\Delta^+$ has then the specific properties that it contains exactly one symbol from $\Gamma_M\times Q_M$ and that it has length $p(|x|)$ always, i.e., we are possibly filling up a string that is too short by the blank symbol~$\VS$.
Configuration sequences of $M$, or runs for short, can be encoded by words from $\#(\Delta^+\#)^*$, or more precisely, from $L_{simple-run}=\#((\Gamma_M\times\{\$\})^*(\Gamma_M\times Q_M)(\Gamma_M\times\{\$\})^*\#)^*$. The latter language can be encoded by a 3-state DFA. However, we will not make use of this language in the following, as it does not fit in the level of the dot-depth hierarchy that we are aiming at.

Let $\Sigma=\{0,1\}$ be the binary target alphabet. 
A letter $a\in \Delta_\#$ is first encoded by a binary word $\hat a$ of length $K=\lfloor \log_2(|\Delta_\#|)\rfloor$, but this is only an auxiliary encoding, used to define the block-encoding 
$$\mathop{enc}(a)=001\hat a[1]1\hat a[2]1\cdots \hat a[K]1$$ of length $L=2K+3$. This block-encoding is extended to words and sets of words as usual. In order to avoid some case distinctions, we assume that $|\Delta_\#|$ is a power of two, so that $\mathop{enc}(\Delta_\#)=001\Sigma 1\Sigma 1\cdots \Sigma 1$. Hence, $\mathop{\overline{enc}}(\Delta_\#)=00\Sigma^{L-2}\setminus \mathop{enc}(\Delta_\#)=00\{a_1b_1a_2b_2\cdots a_Kb_K\mid a_1a_2\cdots a_K\in\Sigma^K\land b_1b_2\cdots b_K\in 1^*0\Sigma^*\}$. Clearly, there are DFAs with $\mathcal{O}(L)$ many states accepting $\mathop{enc}(\Delta_\#)$ and $\mathop{\overline{enc}}(\Delta_\#)$ $(\dagger)$.  In this proof, we will call DFAs with $\mathcal{O}(L\cdot p(|x|))$ many states \emph{small}.
Any encoded word $\mathop{enc}(w)$, with $w\in\Delta_\#^*$, contains the factor $00$ only at positions (minus one) that are multiples of $L$, more precisely: $\mathop{enc}(w)[i]=\mathop{enc}(w)[i+1]=0$ if and only if $i-1$ is divisable by~$L$. This observation allows us to construct small DFAs for $\Sigma^*\mathop{enc}(\Delta_\#)c\Sigma^*$ (for $c\in\{1,01\}$) and for $\Sigma^*\mathop{\overline{enc}}(\Delta_\#)\Sigma^*$, based on $(\dagger)$. 
As shown in the proof of Theorem~3 in~\cite{KMT17}, the language of words that are \underline{not} encodings over $\Delta_\#$ at all is the union of the following languages:
\begin{enumerate}
\item $(1\cup 01)\Sigma^*$,
\item $\Sigma^*\mathop{\overline{enc}}(\Delta_\#)\Sigma^*$, 
\item  $\Sigma^*\mathop{enc}(\Delta_\#)(1\cup 01)\Sigma^*$, and
\item $\Sigma^*00(\bigcup_{i=1}^{L-3}\Sigma^i)=\{w\in\Sigma^*\mid \text{The factor }00\text{ is in the last }L-1\text{ positions}\}$.
\end{enumerate}

Each of these languages can be accepted by small DFAs $A_1, A_2, A_3, A_4$.

Then, we have to take care of the binary words that cannot be encodings of configuration sequences, because the first configuration is not initial.
By our construction, the (unique) initial configuration~$\gamma$ is encoded by
a binary string $\mathop{enc}(\gamma)$ of length  $L\cdot p(|x|)$, i.e., we consider a language $L'$ which is the complement of $\mathop{enc}(\#\gamma\#)\Sigma^*$, the language of all binary strings that do not start with the encoding of the initial configuration. Let $\#\gamma\#=a_1a_2\cdots a_{p(|x|)+2}$.
As we already described non-encodings by automata $A_1$ through $A_4$, instead of  $L'$, we describe $\bigcup_{j=0}^{p(|x|)+2} L_j'$, where
$L_0'=\bigcup_{i=0}^{L\cdot (p(|x|)+2)-1}$ is a finite language (of strings that are too short), $L_j'=\Sigma^{(j-1)L}\mathop{\overline{enc}}(a_j)\Sigma^*$ for $j=1$ to $p(|x|)+2$, describing a violation at symbol $a_j$ of the initial configuration $\gamma$.
Moreover, there are small DFAs $A_5,A_6,\ldots,A_{p(|x|)+7}$ that accept  $L_0',L_1',\ldots, L_{p(|x|)+2}$.


Assuming a unique final state and also assuming that $M$ cleans up the tape after processing, there is a unique final configuration $\gamma_f$ that should be reached. Then, invalidity of a computation with respect to the final configuration can be checked as for the initial configuration, giving us small DFAs $A_{p(|x|)+8},A_{p(|x|)+9},\ldots,A_{2p(|x|)+10}$.

Finally, we want to check the (complement of the) following property of a valid configuration sequence~$\rho\in \#(\Delta^+\#)^*$: any sequence of three letters $a,b,c$ in $\rho$ determines the  letter $f(a,b,c)$ that should be present at a distance of $p(|x|)-1$ to the right. More precisely, we are interested in any factor  $abcvdf(a,b,c)e$ of $\rho$ where $|vd|=p(|x|)-1$.
Different scenarios can occur; we only describe three typical situations in the following.
\begin{itemize}
\item If $a=({a}',\$)$, $b=({b}',\$)$, $c=({c}',\$)$, then $f(a,b,c)=b$. For $d$, we know $d\in \{a'\}\times(Q\cup\{\$\})$ and similarly for $e$, we know $e\in\{c'\}\times(Q\cup\{\$\})$.
\item If $a=({a}',\$)$, $b=\#$, $c=({c}',\$)$, then $f(a,b,c)=\#$. For $d$, we know $d\in \{a'\}\times(Q\cup\{\$\})$ and similarly for $e$, we know $e\in\{c'\}\times(Q\cup\{\$\})$.
\item If $a=({a}',\$)$, $b=({b}',q)$, $c=({c}',\$)$ and if $f_M(q,b')=(p,\hat b',L)$, then $d=(a',p)$, $f(a,b,c)=(\hat b,\$)$, and $e=c$.
 \end{itemize}

We refrain from describing all such situations in detail. Yet with some more sloppiness, we write $\mathop{\overline{enc}}(d(f(a,b,c)e)$ for all situations that do not obey the rules for $df(a,b,c)e$
as tentatively formulated before.
Now, for each triple $a,b,c\in\Delta_\#$, consider the binary language $L_{a,b,c}=\Sigma^*\cdot \mathop{enc}(abc)\cdot \Sigma^{L\cdot (p(|x|)-1)}\cdot \mathop{\overline{enc}}(d(f(a,b,c)e)\cdot \Sigma^*$. This language can be accepted by a small DFA $A_{a,b,c}$.

Altogether, we described $2p(|x|)+10+(|\Delta_\#1)^3$ many languages from $\dd{1}$ such that their union does not yield $\Sigma^*$ if and only if $M$ accepts $x$ using $p(|x|)$ space.
Moreover, for each of the languages, we can build small DFAs.

\section{Conclusion and Open Problems}
\label{section:Conclusion}

We have investigated how the increase in complexity within the dot-depth and the 
Straubing-Th\'{e}rien hierarchies is reflected in the complexity of the \INE problem. We have 
shown the complexity of this problem is already completely determined by the very first 
levels of either hierarchy.

Our work leaves open some very interesting questions and  directions of research. 
First, we were not able to prove containment in $\NP$ for the 
\INE problem when the input automata are allowed to be NFAs accepting a language in the level 
$3/2$ or in the level $1$ of the Straubing-Th\'erien hierarchy. Interestingly, we have shown that 
such containment holds in the case of DFAs, but have shown that the technique we have used to 
prove this containment does not carry over to the context of NFAs. In particular, to show this we 
have provided the first exponential separation between the state complexity 
of general NFAs and partially ordered NFAs. The most immediate open question is if \INE{} for NFAs 
accepting languages in~$\dd{1/2}$,\ $\st{1}$, or $\st{3/2}$ is complete for some level higher up in the polynomial-time 
hierarchy ($\PH$), or if this case is already $\PSPACE$-complete. 
Another tantalizing open question is whether 
one can capture the levels of $\PH$ in terms of the \INE problem when the input 
automata are assumed to accept languages belonging to levels of a sub-hierarchy of $\st{2}$. Such sub-hierarchies 
have been considered for instance in \cite{klima2011subhierarchies}.  

It would also be interesting to have a systematic study of these two well-known subregular hierarchies for related problems like \textsc{Non-universality} for NFAs or \textsc{Union Non-universality} for DFAs. Notice the technicality that \textsc{Union Non-universality} (similar to \INE) has an implicit Boolean operation (now union instead of intersection) within the problem statement, while  \textsc{Non-universality} lacks this implicit Boolean operation. This might lead to a small ``shift'' in the discussions of the hierarchy levels that involve Boolean closure. 
Another interesting hierarchy is the group hierarchy~\cite{Pin98}, where we start with the group languages, i.e., languages
acceptable by automata in which every letter induces a permutation of the state set,
at level 0. Note that for group languages, \INE\  is \NP-complete even for a unary alphabet~\cite{StMe73}.
As $\Sigma^*$ is a group language, the Straubing-Th{\'e}rien hierarchy
is contained in the corresponding levels of the group hierarchy, and hence, 
we get \PSPACE-hardness for level 2 and above in this hierarchy. However, we do not know
what happens in the levels in between.


\bibliography{sample}

\end{document}